\documentclass[conference]{IEEEtran}
\IEEEoverridecommandlockouts

\newif\iftr 
\trfalse
\trtrue 

\if CLASSOPTIONcompsoc
\usepackage[nocompress]{cite}

\else
\usepackage{cite}
\fi

\usepackage{xcolor}
\def\BibTeX{{\rm B\kern-.05em{\sc i\kern-.025em b}\kern-.08em
		T\kern-.1667em\lower.7ex\hbox{E}\kern-.125emX}}

\usepackage{amssymb,amsmath,amsthm}
\usepackage[utf8x]{inputenc}
\usepackage{listings}
\usepackage{color}
\usepackage{epsfig}
\usepackage{url}
\usepackage{amsmath}
\usepackage{cite}
\usepackage{graphicx}
\usepackage{color}
\usepackage{amssymb,amsmath,dsfont}
\usepackage{amsthm}
\usepackage{algorithm}
\usepackage{algpseudocode}
\usepackage{epstopdf}
\usepackage{lipsum}
\usepackage{romannum}

\usepackage{pgfplots}
\pgfplotsset{compat=1.13}
\usepackage{tikz}
\usepackage{graphicx}
\usepackage{breqn}

\usepackage{graphics}

\newtheorem{lemma}{Lemma}

\newtheorem{theorem}{Theorem}

\newtheorem{sub-goal}{Sub-goal}

\usepackage{csquotes}
\usepackage{amsmath}
\usepackage{amssymb}
\usepackage{subcaption}
\usepackage{graphicx}
\usepackage{accents}
\usepackage{mathtools}
\usepackage[normalem]{ulem}
\usepackage{dirtytalk}
\usepackage{mathrsfs}
\usepackage{amsmath}

\graphicspath{figures/}

\usepackage{mathtools}

\usepackage{xspace}

\newcommand{\wmax}{w_{\textit{max}}}
\newcommand{\set}[1]{\left\{{#1}\right\}}
\newcommand{\lmin}{l^{\textit{min}}}
\newcommand{\lmax}{l^{\textit{max}}}
\newcommand{\isdef}{\ensuremath{\overset{\textit{def}}{=}}}

\usepackage[textwidth=1.8cm]{todonotes}
\newcommand{\sref}[1]{Section~\ref{#1}}
\newcommand{\thref}[1]{Theorem~\ref{#1}}
\renewcommand{\wmax}{w_{\scriptsize\textrm{max}}}
\renewcommand{\lmin}{l^{\scriptsize\textrm{min}}}
\renewcommand{\lmax}{l^{\scriptsize\textrm{max}}}
\newcommand{\tot}{\textrm{tot}}
\renewcommand{\isdef}{\ensuremath{\overset{\scriptsize\textit{def}}{=}}}
\newcommand{\lp}{\ensuremath{\left(}}
\newcommand{\rp}{\ensuremath{\right)}}
\newcommand{\lb}{\ensuremath{\left[}}
\newcommand{\rb}{\ensuremath{\right]}}
\newcommand{\lc}{\ensuremath{\left\{}}
\newcommand{\rc}{\ensuremath{\right\}}}

\newcommand{\mif}{\mbox{if }}
\newcommand{\mthen}{\mbox{ then }}
\newcommand{\melse}{\mbox{ else }}
\newcommand{\mand}{\mbox{ and }}
\newcommand{\mwith}{\mbox{ with }}

\newcommand{\motherwise}{\mbox{ otherwise}}
\newcommand{\mymod}{\hspace{-2mm}\mod}
\newcommand{\Nats}{\mathbb{N}}
\newcommand{\Reals}{\mathbb{R}}

\begin{document}
	\pagenumbering{arabic}
	\pagestyle{plain}
	\title{Interleaved Weighted Round-Robin: A Network Calculus Analysis}

	\author{\IEEEauthorblockN{Seyed Mohammadhossein Tabatabaee}
		\IEEEauthorblockA{\textit{EPFL}\\
			Lausanne, Switzerland \\
			hossein.tabatabaee@epfl.ch}
		\and
		\IEEEauthorblockN{Jean-Yves Le Boudec}
		\IEEEauthorblockA{\textit{EPFL}\\
			Lausanne, Switzerland \\
                        jean-yves.leboudec@epfl.ch}
                \and
		\IEEEauthorblockN{Marc Boyer}
		\IEEEauthorblockA{\textit{ONERA/DTIS, University of Toulouse} \\
			F-31055 Toulouse, France\\
			Marc.Boyer@onera.fr}
	}
	

	\maketitle

	\begin{abstract}
    Weighted Round-Robin (WRR) is often used, due to its simplicity, for scheduling packets or tasks. With WRR, a number of packets equal to the weight allocated to a flow can be served consecutively, which leads to a bursty service. Interleaved Weighted Round-Robin (IWRR) is a variant that mitigates this effect. We are interested in finding bounds on worst-case delay obtained with IWRR. To this end, we use a network calculus approach and find a strict service curve for IWRR. The result is obtained using the pseudo-inverse of a function. We show that the strict service curve is the best obtainable one, and that delay bounds derived from it are tight (i.e., worst-case) for flows of packets of constant size. Furthermore, the IWRR strict service curve dominates the strict service curve for WRR that was previously published. We provide some numerical examples to illustrate the reduction in worst-case delays caused by IWRR compared to WRR.

\end{abstract} 
	
	\IEEEpeerreviewmaketitle
	\setcounter{page}{1}

	\section{Introduction}
\label{sec:intro}
Weighted Round-Robin (WRR) is a scheduling algorithm that is often used for scheduling tasks, or packets, in real-time systems or communication networks. 
The capacity is shared between several clients or queues by giving each of them a weight, which is a positive integer, and by providing more service to those with larger weights. Specifically,
%
every queue is visited one after the other, and when a queue $i$ with weight $w_i$ has an emission opportunity, it sends $w_i$ packets, or less if fewer packets are present.
The advantage of WRR is that it is fair and simple. However, the service is bursty because up to $w_i$ packets can be served consecutively for queue $i$, which can cause a large worst-case waiting time for other queues. Interleaved Weighted Round-Robin (IWRR) mitigates
this effect \cite{WRR-ATM}. With IWRR, a queue~$i$ with weight $w_i$ has $w_i$ emission opportunities per round and can send up to one packet at every emission opportunity. In contrast, with WRR, it has one emission opportunity  per round and can send up to $w_i$ packets at every emission opportunity.
 Hence, IWRR spreads out emission opportunities of each queue in a round, which is expected to result in a smoother service and lower worst-case delays. There exist several versions of IWRR; we focus on the simplest one, where queue~$i$ has emission opportunities in the first $w_i$ cycles within a round (see \sref{sec:alg} for a formal description of IWRR and \sref{sec:SOTA:WRR-hist} for WRR variants).

We are interested in delay bounds for the worst case, as is typical in the context of deterministic networking. To this end, a standard approach is network calculus. Specifically, with network calculus, the service offered to a flow of interest by a system is abstracted by means of a service curve. A bound on the worst-case delay is obtained by combining the service curve with an arrival curve for the flow of interest. An arrival curve is a constraint on the amount of data that the flow of interest can send; such a constraint is necessary to the existence of a finite delay bound.
The exact definitions 
are recalled in Section \ref{sec:backg}.

The network calculus approach was applied to WRR in \cite[Sec.~8.2.4]{bouillard_deterministic_2018}, where a \emph{strict} service curve is obtained. As explained in \sref{sec:backg:NC}, a strict service curve is a special case of a service curve hence can be used to derive delay (and backlog) bounds. Our first contribution is to obtain a strict service curve for IWRR. Compared to WRR, the interleaving in IWRR makes the analysis more difficult, 
and the method of proof in \cite{bouillard_deterministic_2018} cannot easily be extended. To circumvent this difficulty, we rely heavily on  
the method of pseudo-inverse, recalled in \sref{sec:backg:NC}. As expected, the IWRR strict service curve dominates that of WRR, hence the resulting delay bounds for IWRR are always less than or equal to those for WRR. 

The strict service curve enables us to obtain delay bounds by using network calculus, but such bounds might not always be tight, i.e., they might not always be equal to worst-cases. This is because the strict service curve is an abstraction of the system. Our second contribution is to show that, for flows with packets of constant sizes, the strict service curve obtained for IWRR provides tight delay bounds. We show that the same result holds for the existing strict service curve of WRR. 
Extending such results to flows with packets of variable sizes is left for further study.

The strict service curve obtained for IWRR has some description complexity, see also Fig.~\ref{fig:exp}. Therefore, we provide simplified lower bounds that can be used, at the expense of sub-optimality, when analytic, closed-form expressions are important.

After giving some necessary background on network calculus and the lower-pseudo inverse technique in Section~\ref{sec:backg}, we describe our system model in \sref{sec:alg}. We describe the state of the art in \sref{sec:soa}. In Section~\ref{sec:Results}, we present our strict service curve for IWRR, the proof of which we present in Section~\ref{sec:mproofs}. In Section~\ref{sec:tightness}, we show that both the IWRR and WRR strict service curves are the best possible and that they give tight delay bounds for a flow with constant packet sizes. We use numerical examples to illustrate the worst-case latency improvement of IWRR over WRR obtained with our method in \sref{sec:numerics}.
\iftr
Proofs of results other than \thref{Thm:1} are in Appendix.
\else
Due to lack of space, proofs of results other than \thref{Thm:1} are in \cite{tabatabaee2020interleaved}.
\fi

	\section{Background}
\label{sec:backg}
\label{sec:backg:NC}
We use the framework of network calculus \cite{le_boudec_network_2001, Changbook,bouillard_deterministic_2018}. A flow is represented by a cumulative arrival function $R\in \mathscr{F}$, where $\mathscr{F}$ denotes the set of wide-sense increasing functions
$f:\mathbb{R}^+ \mapsto \mathbb{R^+} \cup \{+\infty\}$ and $R(t)$ is
the number of bits observed on the flow between times $0$ and $t$. We
say that a flow $R$ has $\alpha\in \mathscr{F}$ as arrival curve if
for all $s \leq t$, $R(t) - R(s) \leq \alpha(t-s)$. A frequently used
arrival curve is $\alpha=\gamma_{r,b}$, defined by $\gamma_{r,b}(t) =
rt+b$ for $t>0$ and $\gamma_{r,b}(t)=0$ for $t= 0$ (token
bucket 
arrival curve, with rate $r$ and burst $b$). An arrival curve $\alpha$
can always be assumed to be sub-additive, i.e., to satisfy $\alpha(s+t)\leq \alpha(s)+\alpha(t)$ for all $s,t$.

For two functions $f$ and $g$ in $\mathscr{F}$, the min-plus convolution is defined by
$
		(f \otimes g)(t) = \inf_{0 \leq s \leq t} \{ f(t-s) + g(s)\}
$.
	An example of min-plus convolution used in this paper is illustrated in Fig.~\ref{fig:minplus}.
	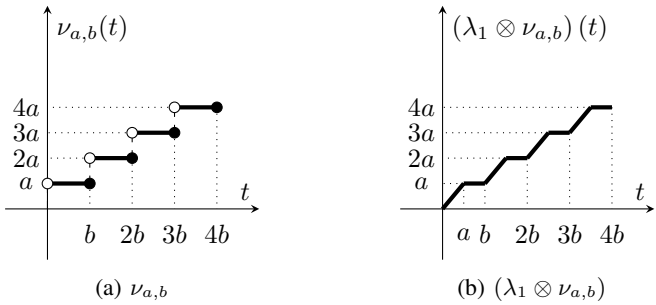
\begin{figure}[htbp]
	\begin{subfigure}[b]{0.4\columnwidth}
		\centering
    		\begin{tikzpicture}

		\pgfplotsset{soldot/.style={color=black,only marks,mark=*}} 				   		   		 \pgfplotsset{holdot/.style={color=black,fill=white,only marks,mark=*}}
		\pgfplotsset{ticks=none}
		\begin{axis}[xlabel=$t$, ylabel=$\nu_{a,b}(t)$,
		xmin=-2,xmax= 10,ymin=-2,ymax=8, axis lines=center, width=1.4\textwidth,
		height=1.4\textwidth,]
		
		\addplot[soldot] coordinates{(2,1)(4,2)(6,3)(8,4)};
		\addplot[holdot] coordinates{(0,1)(2,2) (4,3) (6,4)} ;
		
		\draw[dashed] (axis cs:0,0) -- (axis cs:0,1);
		\draw[ultra thick] (axis cs:0,1) -- (axis cs:2,1);
		\draw[dotted] (axis cs:2,1) -- (axis cs:2,0);

		\draw[dashed] (axis cs:2,1) -- (axis cs:2,2);
		\draw[ultra thick] (axis cs:2,2) -- (axis cs:4,2);
		\draw[dotted] (axis cs:4,2) -- (axis cs:4,0);
		\draw[dotted] (axis cs:2,2) -- (axis cs:0,2);
		
		\draw[dashed] (axis cs:4,2) -- (axis cs:4,3);
		\draw[ultra thick] (axis cs:4,3) -- (axis cs:6,3);
		\draw[dotted] (axis cs:6,3) -- (axis cs:6,0);
		\draw[dotted] (axis cs:4,3) -- (axis cs:0,3);
		
		\draw[dashed] (axis cs:6,3) -- (axis cs:6,4);
		\draw[ultra thick] (axis cs:6,4) -- (axis cs:8,4);
		\draw[dotted] (axis cs:8,4) -- (axis cs:8,0);
		\draw[dotted] (axis cs:6,4) -- (axis cs:0,4);

		\draw (2,-1) node{$b$};
		\draw (4,-1) node{$2b$};
		\draw (6,-1) node{$3b$};
		\draw (8,-1) node{$4b$};

		\draw (-1,1) node{$a$};
		\draw (-1,2) node{$2a$};
		\draw (-1,3) node{$3a$};
		\draw (-1,4) node{$4a$};

		\end{axis}
		\end{tikzpicture}
		\caption{$\nu_{a,b}$}
	\end{subfigure}
\hfill
	\begin{subfigure}[b]{0.4\columnwidth}
		\centering
		\begin{tikzpicture}

		\pgfplotsset{soldot/.style={color=black,only marks,mark=*}} 				   		   \pgfplotsset{holdot/.style={color=black,fill=white,only marks,mark=*}}
		\pgfplotsset{ticks=none}
		\begin{axis}[xlabel=$t$, ylabel=$\left(\lambda_1 \otimes \nu_{a,b}\right)(t) $,
		xmin=-2,xmax=10,ymin=-2,ymax=8, axis lines=center, width=1.4\textwidth,
		height=1.4\textwidth,]

\draw[ultra thick] (axis cs:0,0) -- (axis cs:1,1);
\draw[ultra thick] (axis cs:1,1) -- (axis cs:2,1);
\draw[dotted] (axis cs:1,1) -- (axis cs:0,1);
\draw[dotted] (axis cs:1,1) -- (axis cs:1,0);
\draw[dotted] (axis cs:2,1) -- (axis cs:2,0);

\draw[ultra thick] (axis cs:2,1) -- (axis cs:3,2);
\draw[ultra thick] (axis cs:3,2) -- (axis cs:4,2);
		\draw[dotted] (axis cs:4,2) -- (axis cs:4,0);
\draw[dotted] (axis cs:3,2) -- (axis cs:0,2);

\draw[ultra thick] (axis cs:4,2) -- (axis cs:5,3);
\draw[ultra thick] (axis cs:5,3) -- (axis cs:6,3);
		\draw[dotted] (axis cs:6,3) -- (axis cs:6,0);
\draw[dotted] (axis cs:5,3) -- (axis cs:0,3);

\draw[ultra thick] (axis cs:6,3) -- (axis cs:7,4);
\draw[ultra thick] (axis cs:7,4) -- (axis cs:8,4);
		\draw[dotted] (axis cs:8,4) -- (axis cs:8,0);
\draw[dotted] (axis cs:7,4) -- (axis cs:0,4);
		
		\draw (1,-1) node{$a$};
\draw (2,-1) node{$b$};
\draw (4,-1) node{$2b$};
\draw (6,-1) node{$3b$};
\draw (8,-1) node{$4b$};

\draw (-1,1) node{$a$};
\draw (-1,2) node{$2a$};
\draw (-1,3) node{$3a$};
\draw (-1,4) node{$4a$};

		\end{axis}
		\end{tikzpicture}
				\caption{$(\lambda_1 \otimes \nu_{a,b})$}
	\end{subfigure}
	\caption{\sffamily \small Left: the stair function $\nu_{a,b}\in\mathscr{F}$ defined for $t\geq 0$ by $\nu_{a,b}(t)=a\left\lceil \frac{t}{b}\right\rceil$. Right: min-plus convolution of $\nu_{a,b}$ with the function $\lambda_1\in\mathscr{F}$ defined by $\lambda_1(t)=t$ for $t\geq0$, when $a\leq b$. The discontinuities are smoothed, and replaced with a unit slope.}
	\label{fig:minplus}
\end{figure}

Consider a system $S$ and a flow through $S$ with input and output functions $R$ and $R^*$ and let $\beta\in\mathscr{F}$. We say that the system $S$ offers $\beta$ as a service curve to the flow if $R^*\geq R\otimes \beta$, which often means that for every $t\geq 0$ there exists some $s\leq t$ such that
$R^*(t)\geq R(s)+\beta(t-s)$ \cite[Sec.~3.2.2]{bouillard_deterministic_2018}.
We say that system $S$ offers a \emph{strict} service curve $\beta\in\mathscr{F}$ to the flow if $R^*(t) - R^*(s) \geq \beta(t-s)$ whenever $(s,t]$ is a backlogged period (i.e., $R^*(\tau)>R(\tau)$ for all $\tau$ such that $s<\tau\leq t$). If $\beta$ is a strict service curve, then it is a service curve, but the converse is not always true
\cite[Section 1.3]{le_boudec_network_2001}. A frequently used service curve
is the rate-latency function $\beta_{r,T}$ that is the function in $\mathscr{F}$ defined by $
		\beta_{r,T}(t) = r[t-T]^+$, where we use the notation $[x]^+=\max\set{x,0}$. Saying that a system offers a service curve $\beta_{r,T}$ to a flow expresses that the flow is guaranteed a service rate $r$, except for possible interruptions that might impact the delay by at most $T$. Saying that a system offers a \emph{strict} service curve $\beta_{r,T}$ to a flow expresses that the flow is guaranteed a service rate $r$, except for possible interruptions that might not exceed $T$ in total per backlogged period. A strict service curve $\beta$ can always be assumed to be super-additive, i.e., to satisfy $\beta(s+t)\geq \beta(s)+\beta(t)$ for all $s,t$ (otherwise, it can be replaced by its super-additive closure \cite[Prop. 5.6]{bouillard_deterministic_2018}).

Assume that a flow, constrained by arrival curve $\alpha$,
traverses a system that offers a service curve $\beta$ to the flow
and that respects the ordering of the flow (FIFO per-flow). The delay of the flow is upper bounded by $h(\alpha,\beta)$ (horizontal deviation), defined by
	\begin{equation}\label{eq:horiz}
		h(\alpha,\beta) = \sup_{t \geq 0} \{ \inf \{ d \geq 0 | \alpha(t) \leq \beta(t + d)\}\}
	\end{equation}

Our technique of proof uses the lower pseudo-inverse. The lower pseudo-inverse $f^{\downarrow}$ of a function $f \in \mathscr{F}$ is defined by
\begin{equation}\label{def:lsi}
	f^{\downarrow}(y) = \inf \{x | f(x) \geq y \} = \sup \{ x | f(x) < y \}
	\end{equation}
We use the following property from \cite[Sec. 10.1]{liebeherr2017duality}:
	\begin{equation} \label{lem:lsi}
\forall x,y \in \mathbb{R}^+, y \leq f(x) \Rightarrow x \geq f^{\downarrow}(y)
\end{equation}
%

%

    \section{System Model} \label{sec:alg}
We consider a weighted round-robin subsystem that serves $n$ input flows, has one queue per flow, and uses a weighted round-robin algorithm (described later) to arbitrate between flows. The weighted round-robin subsystem is itself placed in a larger system, and can compete with other queuing subsystems.
For example, consider the case of a constant-rate server with several priority levels, without preemption, and where the weighted round-robin subsystem is at a priority level that is not the highest. Assuming some arrival curve constraints for the higher priority traffic, the service received by the entire weighted round-robin subsystem can be modelled using a strict service curve \cite[Section 8.3.2]{bouillard_deterministic_2018}.

%

This motivates us to assume that the aggregate of all flows in the weighted round-robin subsystem receives a strict service curve, say $\beta\in \mathscr{F}$ that we call ``aggregate strict service curve". If the weighted round-robin subsystem has exclusive access to a transmission line of rate $c$, then $\beta(t)=ct$ for $t\geq 0$.
We assume that $\beta(t)$ is finite for every (finite) $t$ and, without loss of generality, we assume $\beta$ to be super-additive. Furthermore, we need an additional technical assumption, primarily for establishing the tightness result: we assume that $\beta$ is Lipschitz-continuous, i.e., there exists a constant $K>0$ such that $\frac{\beta(t)-\beta(s)}{t-s}\leq K$ for all $0\leq s<t$; this does not appear to be a restriction as the rate at which data is served has a physical limit.

%


The arbitration algorithm assumed in this paper is IWRR, shown in Algorithm \ref{alg:IWRR}. When a packet of flow $i$ enters the weighted round-robin subsystem, it is put into queue $i$. The weight of flow $i$ is 
$w_i$. IWRR runs an infinite loop of \emph{rounds}. In one round, each queue $i$ has $w_i$ emission opportunities; one packet can be sent during one emission opportunity. The inner loop defines a \emph{cycle}, where each queue is visited but only those with a weight not smaller than the cycle number have an emission opportunity. The \texttt{send} instruction is assumed to be the only one with a non-null duration. Its actual duration depends on the packet size but also on the amount of service available to the entire weighted round-robin subsystem. See Figure~\ref{fig:Scheduling532} for an illustration.

\begin{figure}[htbp]
  \centering
  \includegraphics[width=\linewidth]{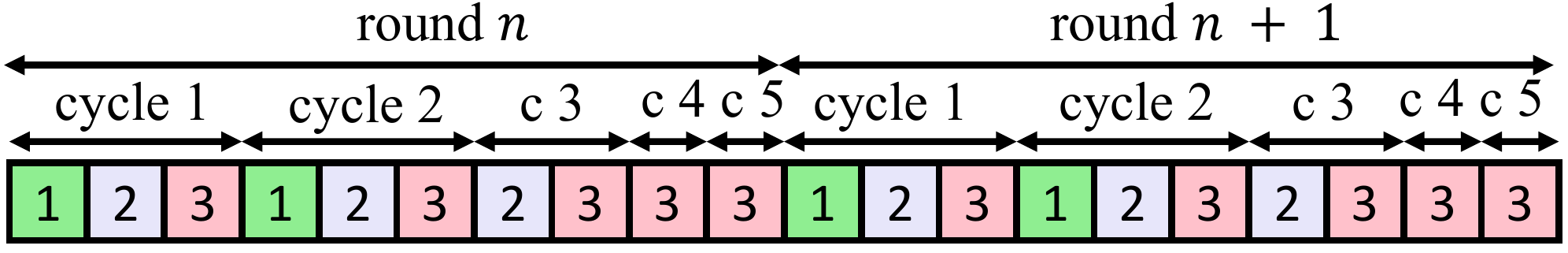}
  \caption{\sffamily \small Emission opportunities on two successive rounds for IWRR with three flows and $w_1=2, w_2=3, w_3=5$. Mind that
this is not the temporal behaviour: each opportunity can lead to an
empty interval if the queue is empty at this time. Furthermore, the duration of
each non-empty interval depends on the packet size and the aggregate service available (we do not assume constant rate service).}
  \label{fig:Scheduling532}
\end{figure}


%

\begin{algorithm}[htbp]
\textbf{Input:}{ Integer weights $w_1 \leq w_2 \leq .. \leq w_n$}
\caption{Interleaved Weighted Round-Robin}
\begin{algorithmic}[1]
\State $w_{\max} = \max\{w_1,..,w_n\}$
\While{True}
\Comment{A round starts.}
	\For{$C\leftarrow 1$ to $w_{\max}$}
\Comment{A cycle starts.}
		\For{$i \leftarrow 1 $ to $n$}
			\If{$C \leq w_i$}
				\If{(\textbf{not} empty($i$))}\\
\Comment{A service for queue $i$.}
					\State \texttt{print}(\text{now},$i$);\label{alg:service}
					\State \texttt{send}(head($i$));
					\State\texttt{removeHead(}$i$);
				\EndIf
			\EndIf
		\EndFor
	\EndFor \Comment{A cycle finishes.}
\EndWhile \Comment{A round finishes.}
\end{algorithmic}
\label{alg:IWRR}
\end{algorithm}

Here, we use the context of communication networks, but the results equally apply to real-time systems: Simply map flow to task, packet to job, packet size to
job execution time and strict service curve to ``delivery curve" \cite{4617308,858698}.


    \section{State of the art}
\label{sec:soa}

\label{sec:SOTA:WRR-hist}

One of the first use of round-robin scheduling in the network context
appeared in \cite{RR-86}, with a fairness objective, i.e., a fair way to
share the bandwidth between sessions.  It is also mentioned in
\cite{FQ-Nagle} as a way to implement ``fair queueing''.

The term ``Weighed Round-Robin'' was coined in \cite{WRR-ATM}
as a generalisation of round-robin to share the bandwidth ``in
proportion to prescripted weights'' in the context of ATM (i.e., with
constant size packets). Two versions of the algorithm are presented in \cite{WRR-ATM}. The former is presented in Algorithm~\ref{alg:IWRR}:
at cycle $C$ (with $C$ between $1$ and $w_{\max}$), only
flows with weight $w_i \geq C$ can emit one packet. We call this version IWRR.
The latter version assumes that there exists for each flow $i$ a bit-list of length
$\wmax$, $o_i \in \set{0,1}^{\wmax}$, such that
$w_i=\sum_{k=1}^{\wmax}o_i[k]$.  A flow $i$ can emit a packet at cycle
$C$ only if $o_i[C]=1$. A strategy is given to build these vectors
in \cite{WRR-ATM} and is refined with fairness objectives in \cite{WRR-SlotShaping}.
Call LIWRR (list-based IWRR) this version.

IWRR is modified into WRR/SB in \cite{WRR-SB} to enable some flow to
send slightly more packets than permitted in a cycle, and to decrease
accordingly at the next cycle.

As mentioned in \sref{sec:intro}, plain WRR (which we simply call ``WRR")
enables each flow $i$ to send up to $w_i$ packets every time it is
selected \cite{Multiclass-RR-Conf}. A ``Multiclass WRR'' is also defined in
\cite{Multiclass-RR-Conf}. Surprisingly, the authors of
\cite{Multiclass-RR-Conf} were not aware of \cite{WRR-ATM} and have re-invented LIWRR.
Note that even if WRR was designed for packets of constant size, it
has been applied in network of variable size packets such as Ethernet
\cite[Sec.~8.6, Sec.~8.6.8.3, Sec.~37]{802.1Q-2018}, in request balancing in
cloud infrastructures \cite{HoneyBee-CloudLoadBalancing}, in the
LinuxVirtualServer scheduling \cite{LVS-WRR}, in network of chip
\cite{NC-Wormhole-WRR-2009}, and so on. In fact, looking
for expression ``weighted round-robin'' in the title or abstracts of
papers index by Scopus returns  more than 400 entries (March 2020),
and Google references more than 4000 patents with this expression
(March 2020).
%
%
%
Unfortunately, when authors refer to WRR, they often do not
explicit which version of WRR it is. 

\label{sec:SOTA:WRR-WCD}

A WRR server is also a latency-rate server, with latency
and rates given in \cite{LR-TON-98} for 
packets of constant size. The latency result is
generalised to LIWRR in \cite{LIWRR-Latency}. Even if the notion of
latency-rate server is very close to the one of a service curve
$\beta_{r,T}$ in network calculus, both notions are slightly different,
and results cannot be directly imported from one theory to the other
\cite{LR-GR}. 
In \cite{NC-Wormhole-WRR-2009}, the authors consider a Network on Chip (NoC), with WRR
arbitration at the flit level. A flit is the elementary data
  unit of the NoC, one flit is sent per CPU/NoC cycle. Assuming that the weights are such that packets are never fragmented by the arbiter,
a strict service curve $\beta_{R_i,T_i}$ for flow $i$ is
found, with
$R_i=\frac{w_i}{\sum_{k}w_k}$, $T_i=\sum_{j\neq i}w_j$.

WRR arbitration in an Ethernet switch is also considered
in \cite{NC-EthSw-WRR}, with the assumption that all flows of an
output ports have the same constant packet size. It then computes, in
the network calculus framework, a residual service with service curve
$\beta_{R_i,T_i}$ with $R_i=\frac{w_i}{\sum_{k}w_k}C$,
$T_i=\frac{\sum_{j\neq i}w_j}{C}$, where $C$ is the link rate. We
assume that the missing packet size in the $T_i$ term was a typo.
This network calculus result on conventional WRR arbitration in
Ethernet is refined in \cite{WRR-NC-Avionic}, considering packets of
variable size, leading to  residual service with strict service curve
$\beta_{R_i,T_i}$ with $R_i =\frac{w_i\lmin_i}{w_i\lmin_i + \sum_{j\neq i}w_j\lmax_j}C$ and $T_i=\frac{\sum_{j\neq i}w_j\lmax_j}{C}$
(cf. eq.~(1) and (2) in
\cite{WRR-NC-Avionic}) where $\lmin_i, \lmax_i$ are, respectively, lower
and upper bounds on the size of the packets in the flow $i$. It refines this result by subtracting the part of the bandwidth
not used by interfering flows (considering their arrival curves).

Observe that computing a residual service with a $\beta_{R,T}$ curve is
pessimistic as it assumes that, once the worst latency is payed,
each packet is served with the long-term residual rate. Whereas, in
reality, each packet, when it is selected for emission, is transmitted
at full link speed up to completion. A residual
service for the conventional WRR with a curve that is an alternation of full
services and plateaus is given in
\cite[Sec.~8.2.4]{bouillard_deterministic_2018}.
This effect of ``full speed up to completion'' can also be captured
when computing the local delay of a server with $\beta_{R,T}$ service
curve \cite{NC-Packet-Delay-Ratency}.


\section{Strict Service Curves for IWRR} \label{sec:Results}
Our first result is a strict service curve for IWRR that, as we show in \sref{sec:tightness}, is the best possible. We compare it to WRR and also give simpler, lower approximations.
\begin{theorem}[Strict Service Curve of IWRR]\label{Thm:1}
Let $S$ be a server shared by $n$ flows that uses IWRR as explained in \sref{sec:alg}, with weight $w_i$ for flow $i$. Recall that the server offers a 
strict service curve $\beta$ to the aggregate of the $n$ flows. For any flow $i$, $\lmin_i$ [resp.$\lmax_i$] is a lower [resp. upper] bound on the packet size.

Then, $S$ offers to every flow $i$ a strict service curve $\beta_i$ given by $\beta_i(t)=\gamma_i(\beta(t))$ with
\begin{align}
\label{eq:gamma}
\gamma_i  &= \lambda_1 \otimes U_i
\\
U_i(x) &\isdef \sum_{k=0}^{w_i - 1}\nu_{\lmin_i,L_{\tot}}\left(\lb x - \psi_i(k\lmin_i)\rb^+\right)
\label{eq:u}
\\
\label{eq:Ltot}
L_{\tot} &=w_i\lmin_i + \sum_{j,j \neq i} w_j\lmax_j
\\
\label{eq:psi}
\psi_i(x) &\isdef x + \sum_{j,j \neq i} \phi_{i,j}\left(\left\lfloor \frac{x}{\lmin_i} \right\rfloor\right)\lmax_j
\\
\label{eq:phi}
\phi_{i,j}(x) &\isdef \left\lfloor \frac{x}{w_i} \right\rfloor w_j +  \left[w_j - w_i\right]^+ \nonumber\\&+ \min(x \mymod  w_i+ 1,  w_j)
\end{align}
%
%
In the above, $\nu_{a,b}$ is the stair function, $\lambda_1$ is the unit rate function and $\otimes$ is the min-plus convolution, all are described in Fig.~\ref{fig:minplus}.

Furthermore, $\beta_i$ is super-additive.
\end{theorem}
The proof is in \sref{sec:mproofs}.
See Fig.~\ref{fig:exp} for some illustration of $\beta_i$. Observe that $\gamma_i$ in \eqref{eq:gamma} is the strict service curve obtained when the aggregate strict service curve is $\beta=\lambda_1$ (i.e., when the aggregate is served at a constant, unit rate). In the common case where $\beta$ is equal to a rate-latency function, say $\beta_{c,T}$, we have $\beta_i(t)=\gamma_i(c(t-T))$ for $t\geq T$ and $\beta_i(t)=0$ for $t\leq T$, namely, $\beta_i$ is derived from $\gamma_i$ by a rescaling of the $x$ axis and a right-shift.

As mentioned in \sref{sec:backg}, any strict service curve that is not super-additive can be replaced by its super-additive closure. The last statement in the theorem guarantees that this does not occur here.

%
We now compare to WRR. The best known service curve for (non-interleaved) WRR is given in
\cite[Sec. 8.2.4]{bouillard_deterministic_2018} and is
	\begin{equation}\label{eq:WRR}
\beta_i'(t) =  (\lambda_1 \otimes \nu_{q_i,L_{\tot}})  \left(\lb \beta(t) - Q_i \rb ^+\right)
\end{equation}	
with $q_i =  w_i\lmin_i$ and $Q_i = \sum_{j,j \neq i} w_j\lmax_j$.   In \sref{sec:tightness}, we show that $\beta_i'(t)$ is indeed the best possible strict service curve for WRR. Furthermore, it is dominated by the strict service curve for IWRR:

\begin{theorem} \label{thm:WRRvsIWRR} With the assumptions in Theorem~\ref{Thm:1} and in \eqref{eq:WRR}:
	\begin{equation}
	\beta_i' \leq \beta_i
	\end{equation}	
\end{theorem}
The proof is in \iftr Appendix. \else \cite{tabatabaee2020interleaved}. \fi Fig.~\ref{fig:exp} illustrates how the strict service curve for IWRR improves on that for WRR, by providing a smoother, and generally larger, service.


\begin{figure*}[htbp]
		\centering
		\title={Illustration of Theorems \ref{Thm:1}, \ref{thm:WRRvsIWRR}, and \ref{theo:sc}}
		 \includegraphics[width=\linewidth]{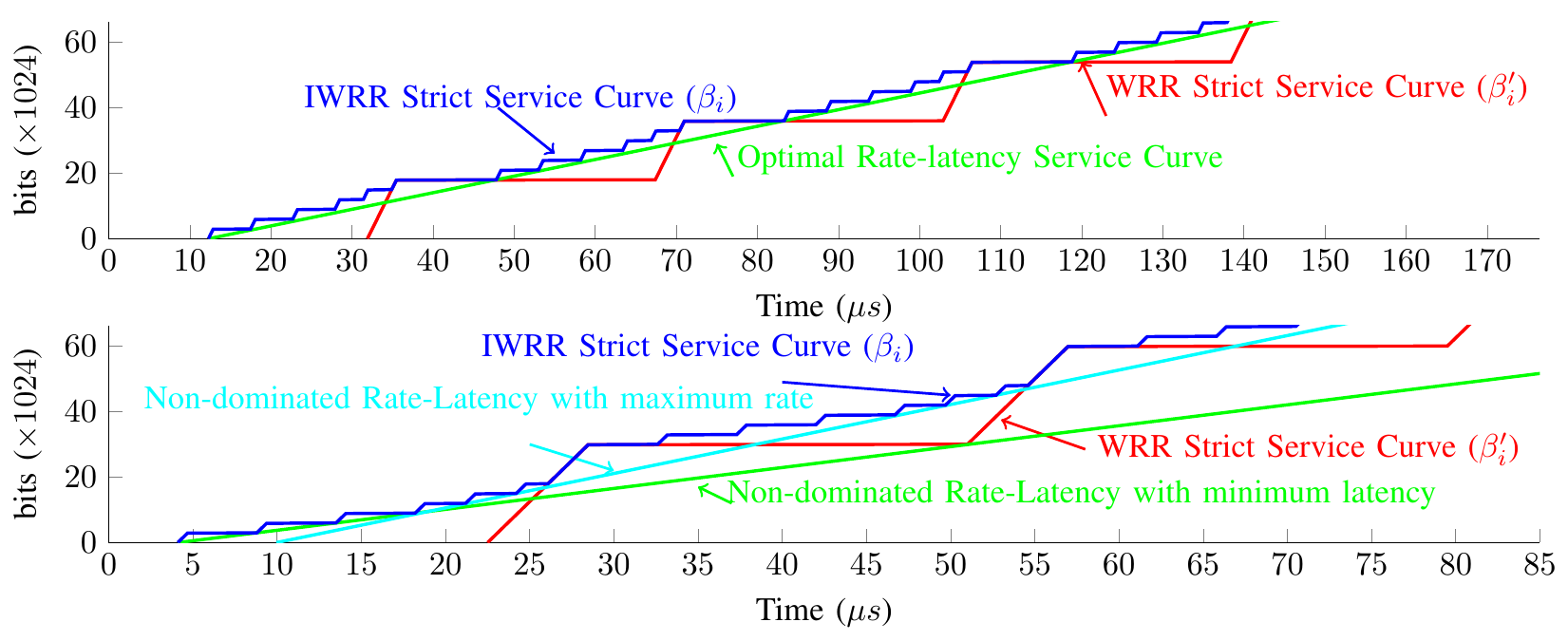}
		\caption{\sffamily \small
Strict service curves obtained in \sref{sec:Results} for an example with four input flows, weights = $\{4,6,7,10\}$, $\lmin = \{4096, 3072, 4608, 3072\} \text{ bits}$, $\lmax = \{8704 , 5632,6656, 8192\} \text{ bits}$ and $ \beta(t)=ct$ with $c=10$Mb/s (i.e., the aggregate of all flows is served at a constant rate). The figure shows the IWRR service curve $\beta_i$ and the WRR strict service curve $\beta_i'$ for two of the flows; it also shows the non-dominated rate-latency strict service curves $\beta_{r_0^*,T_0^*}$ and $\beta_{r_{k^*}^*,T_{k^*}^*}$ of \thref{theo:sc} (in the top panel both are equal).}

		\label{fig:exp}
	\end{figure*}

The service curve found in Theorem \ref{Thm:1} is the best possible one but has a complex expression. If there is interest in a simpler expression, any lower bounding function is a strict service curve; in particular, the strict service curve $\beta_i'$ for WRR is also a valid, though suboptimal, strict service curve for IWRR. There is often interest in service curves that are rate-latency functions. Observe that, if the aggregate service curve $\beta$ is a rate-latency function, then replacing $\gamma_i$ by a rate-latency lower-bounding function also yields a rate-latency function for $\beta_i$, and vice-versa. Therefore, we are interested in rate-latency functions that lower bound $\gamma_i$.

Among all of such these, there is not a single best one, as some have a smaller latency while others have a larger rate. We say that a rate-latency function $\beta_{r,T}$ that lower bounds $\gamma_i$ is non-dominated if there is no other rate latency function $\beta_{r',T'}$ that lower bounds $\gamma_i$ and dominates $\beta_{r,T}$, i.e., such that $r'\geq r$ and $T'\leq T$. The following result gives all such non-dominated rate-latency functions. Let  $r^* = \frac{q_i}{L_{\tot}}=\frac{ w_i\lmin_i}{L_{\tot}}$, $r_{w_i - 1} = 1$, and
\begin{align}
  r_k &= \frac{\lmin_i}{\psi_i((k + 1)\lmin_i) - \psi_i(k\lmin_i)}, \;0\leq k < w_i - 1
  \label{eq:defrk} \\
  k^*&= \min \{0 \leq k < w_i ~ | ~ r_k \geq r^*\}  \label{eq:defk*} \\
  r_k^*&= \min(r_k , r^*),\;  0\leq k\leq k^* 
  \label{eq:defrk*}
  \end{align}  
\begin{theorem} \label{theo:sc} 
With the assumptions in Theorem~\ref{Thm:1} and the definitions  \eqref{eq:defrk}-\eqref{eq:defrk*}, a rate-latency function $\beta_{r,T}$ lower bounds $\gamma_i$ and is non-dominated if and only if $r = r_{k^*}^*$ and $T =\psi_i(k^*\lmin_i) - \frac{k^*\lmin_i}{r}$, or $r_{k-1}^* \leq r < r_k^*$ and $T =  \psi_i(k\lmin_i) - \frac{k\lmin_i}{r}$ for some integer $k$ with $0<k\leq k^*$.
Among all such rate-latency functions, the one with lowest latency is $\beta_{r_0^*,T_0^*}$ and the one with largest rate is $\beta_{r_{k^*}^*,T_{k^*}^*}$.
 \end{theorem}


The proof is in \iftr Appendix and \else \cite{tabatabaee2020interleaved}. \fi  Fig~\ref{fig:exp} illustrates $\beta_{r_0^*,T_0^*}$ and $\beta_{r_{k^*}^*,T_{k^*}^*}$ in some examples. Observe that $k\mapsto r_k^*$ is wide-sense increasing with $k$ for $0\leq k\leq k^*$, but the values of $r_k^*$ are not necessarily all distinct.
 It can also occur that $k^*=0$ (as in the top panel of Fig.~\ref{fig:exp}); 
 in which case, there is one optimal rate-latency service curve. In general, however, this does not occur, and a simple lower bounding approximation can be obtained with the supremum of all non-dominated rate-latencies, which can be shown is equal to $\max\lp\beta_{r_0^*,T_0^*},\ldots, \beta_{r_{k^*}^*,T_{k^*}^*}\rp$. When $\beta$ is a rate-latency function, this provides a convex piecewise linear function that has several good properties \cite[Sec.~4.2]{bouillard_deterministic_2018}. 

%
%
%

\section{Proof of Theorem \ref{Thm:1}} \label{sec:mproofs}
The idea of proof is as follows. We consider a backlogged period $(s,t]$ of flow of interest $i$, and we let $p$ be the number of packets of flow $i$ that are entirely served during this period. For every other flow $j$, the number of packets that are entirely served is upper bounded by a function of $p$, given in Lemma \ref{lem:numg}. $p$ is also upper bounded by a function of the amount of service received by flow $i$ in Lemma \ref{prop:p}. Combining these two results in an implicit inequality for the total amount of service \eqref{eqn:TotService}. By using the technique of pseudo-inverse, this inequality is inverted and provides a lower bound for the amount of service received by the flow of interest.

\subsection{Key Variables and Basic Properties} \label{ssec:t}
 Let $(s,t]$ be a backlogged period of flow $i$. Let
 $(\tau_k,fl_k)$ be couples of (instant,flow), printed at line \ref{alg:service} of Algorithm \ref{alg:IWRR}. 
 Note that 
 $\tau_k<\tau_{k+1}$ as the \texttt{send} instruction has a non-null duration (because the aggregate service curve $\beta$ is Lipschitz continuous). Let $\sigma(0), \sigma(1), \ldots$ be the sequence of service opportunities for flow $i$ at or after $s$, i.e., $\sigma(0)= \min \{ m \;|\; \tau_m \geq s , fl_m = i\}$ and $\sigma(k)= \min\{ m \;|\; \tau_m > \tau_{\sigma(k-1)},fl_m = i\}$.
 The $k$th service opportunity for flow $i$ occurs at time $\tau_{\sigma(k-1)}$; we say that it is ``complete" if $\tau_{\sigma(k-1)+1}\leq t$, i.e., the interval taken by this service is entirely in $[s,t]$. Let $p\geq 0 $ be the number of complete service opportunities. Observe that it is possible that $p=0$, and it might happen that  $\tau_{\sigma(p)}<t$ or $\tau_{\sigma(p)}\geq t$. 

 \begin{center}
\begin{tikzpicture}
\begin{axis}[%
width=4in,
axis x line=center,
axis y line=none,
xmin=0,xmax=20,ymin=0, ymax=1,
xtick={0.1, 2, 4.5, 7, 10, 13, 17, 19},
xticklabels={$s$,$\tau_{\sigma(0) - 1 }$,$\tau_{\sigma(0)}$, $\tau_{\sigma(0) + 1} $,  $\tau_{\sigma(p-1) } $, $\tau_{\sigma(p-1) + 1 } $,  $t$ , $\tau_{\sigma(p)}$}
]

\draw [dotted, <->] (axis cs:4.5,0.05) -- (axis cs:7,0.05) node[pos=0.5, above] {\sffamily \small flow $i$ is served};

\draw [dotted, <->] (axis cs:10,0.05) -- (axis cs:13,0.05) node[pos=0.5, above] {\sffamily \small flow $i$ is served};
\draw [dotted] (axis cs:7.5,0.025) -- (axis cs:9.5,0.025);

\draw [dotted] (axis cs:13.5,0.025) -- (axis cs:16.5,0.025);

\end{axis}
\end{tikzpicture} 
\begin{tikzpicture}
\begin{axis}[%
width=4in,
axis x line=center,
axis y line=none,
xmin=0,xmax=20,ymin=0, ymax=1,
xtick={0.1, 2, 4.5, 7, 10, 13, 16, 17.25, 19},
xticklabels={$s$,$\tau_{\sigma(0) - 1 }$,$\tau_{\sigma(0)}$, $\tau_{\sigma(0) + 1} $,  $\tau_{\sigma(p-1) } $, $\tau_{\sigma(p-1) + 1 } $ , $\tau_{\sigma(p)}$,  $t$ , $\tau_{\sigma(p) + 1}$}
]

\draw [dotted, <->] (axis cs:4.5,0.05) -- (axis cs:7,0.05) node[pos=0.5, above] {\sffamily \small flow $i$ is served};
\draw [dotted, <->] (axis cs:10,0.05) -- (axis cs:13,0.05) node[pos=0.5, above] {\sffamily \small flow $i$ is served};
\draw [dotted, <->] (axis cs:16,0.05) -- (axis cs:19,0.05) node[pos=0.5, above] {\sffamily \small flow $i$ is served};

\draw [dotted] (axis cs:7.5,0.025) -- (axis cs:9.5,0.025);
\draw [dotted] (axis cs:13.5,0.025) -- (axis cs:15.5,0.025);

\end{axis}

\end{tikzpicture} 
\end{center}



In each service of flow $i$, during a backlogged period, it sends one packet with a length $\geq \lmin_i$, thus, for all $k=0\ldots (p-1)$,
$
R_i^*(\tau_{\sigma(k + 1)}) - R_i^*(\tau_{\sigma(k)}) \geq \lmin_i
$, 
therefore
\begin{equation}\label{eqn:Min}
R_i^*(\tau_{\sigma(p)}) - R_i^*(\tau_{\sigma(0)}) \geq p\lmin_i
\end{equation}



\subsection{Amount of Service to Other Flows}

In order to upper bound the number of emission opportunities for another flow  $j$, 
we first find an expression, in Lemma \ref{lem:NVg}, for the number of emission opportunities for flow $j$ between two consecutive emission opportunities for flow~$i$. Lemma \ref{lem:NVg2} then finds an upper bound on the number of emission opportunities for flow $j$ in $(s , \tau_{\sigma(p)})$, as a function of the cycle number (variable $C$ in Algorithm~\ref{alg:IWRR}) at $\tau_{\sigma(0)}$. Lastly, Lemma~\ref{lem:numg} maximizes the previous upper bound over all values of $C$.
\begin{lemma} \label{lem:NVg}
The number of emission opportunities for flow $j\neq i$ between two consecutive emission opportunities for flow $i$, given that the latter emission opportunity for flow $i$ occurs at cycle $C$, is equal to $q_{i,j}(C)\isdef$
\begin{equation}
\begin{cases}
   0&\mif 1 < C \leq w_i \mand w_j < C\\
   1&\mif 1 < C \leq w_i \mand w_j \geq C\\
    \left[w_j - w_i\right]^+ + 1 &\mif C= 1
\end{cases}
\label{eq:lem-l2}
\end{equation}
\end{lemma}
 %
%
%
%
%
\begin{proof}
	
According to Algorithm \ref{alg:IWRR}, flow $i$ has emission opportunities only in the first $w_i$ cycles of each round. Both emission opportunities are either in the same round (Case 1) or in two consecutive rounds (Case 2). As $C$ is the cycle number for the second emission opportunity for flow $i$, Case 1 can occur only when $1 < C \leq w_i$, and Case 2 can occur when $C=1$. For Case~1, we further differentiate between $w_j<C$ and $w_j\geq C$.

\textbf{Case 1a:} $1 < C \leq w_i$ and $w_j< C$: Queue $j$ does not have an  emission opportunity in cycle $C$ because $w_j< C$. Also, we must have $w_j<w_i$, thus queue $j$ does not have an emission opportunity after $i$ in cycle $C-1$. Hence, $q_{i,j}(C)=0$.
	
\textbf{Case 1b:} $1 < C \leq w_i$ and $w_j\geq C$: If $w_j>w_i$,
then queue $j$ has an emission opportunity after queue $i$ in cycle $C - 1$. If $w_j=w_i$, then queue $j$ has an emission opportunity before $i$ in cycle $C$, or after $i$ in cycle $C-1$. Else, $C\leq w_j<w_i$ and queue $j$ has an emission opportunity in cycle $C$, before $i$. In all cases, $q_{i,j}(C)=1$.
 	
\textbf{Case 2:} $C  = 1$:	
	The first emission opportunity for $i$ is in the last cycle of a round that includes $i$ (cycle $w_i$). If $w_j>w_i$, then queue $j$ has an emission opportunity in the rest of cycle $w_i$ and also has emission opportunities during the next $(w_j - w_i)$ cycles of the last round. In this case, $q_{i,j}(C)=w_j - w_i + 1$, which is also the value in the last line of \eqref{eq:lem-l2}. Else if $w_j= w_i$, queue $j$ has an emission opportunity before $i$ in this cycle or after $i$ in cycle $w_i$ of the first round, thus $q_{i,j}(C)=1$, which is also the value in the last line of \eqref{eq:lem-l2}. Else, $w_j<w_i$ and queue $j$ has an emission opportunity before $i$ in this cycle. Here too, $q_{i,j}(C)=1$, 
 the value in the last line of \eqref{eq:lem-l2}.
\end{proof}

\begin{lemma} \label{lem:NVg2}
	
	The number of emission opportunities for flow $j \neq i$ in $(s , \tau_{\sigma(p)})$, for any backlogged period $(s,t]$ of flow $i$ with $p$ complete services, given that the first service starts at cycle number $C$ (cycle number at time $\tau_{\sigma(0)}$) is upper bounded by
	\begin{equation}\label{eqn:lem21}
	 q'_{i,j}\left(C, p \right) \isdef \sum_{k=0}^{p} q_{i,j} \left (\left( C + k - 1\right) \mymod  w_i + 1 \right)
	\end{equation}
	Also, let $C'(p)$ be the cycle number at $\tau_{\sigma(p)}$. Then,
	\begin{equation}\label{eqn:lem22}
	C'(p) =  \left( C + p - 1\right) \mymod  w_i + 1
	\end{equation}
\end{lemma}

\begin{proof}
	
	By induction on $p$.
	
	\textbf{Base Case:} $ p = 0$
	
	In this case, $q'_{i,j}\left(C , 0 \right)$ is the number of emission opportunities for flow $j$ between two consecutive emission opportunities for flow $i$ that by Lemma \ref{lem:NVg}, is equal to $q_{i,j}(C)$. As $1 \leq C \leq w_i$, $\left( C - 1\right) \mymod   w_i + 1  = C$ thus $q_{i,j}(C)=q_{i,j} \left (\left( C  - 1\right) \mymod  w_i + 1 \right)$. This shows \eqref{eqn:lem21}.   
	Also, by definition, $C'(0)=C$; using again $\left( C - 1\right) \mymod   w_i + 1  = C$ shows that \eqref{eqn:lem22} holds.
	
	\textbf{Induction step:}
	
	We assume that \eqref{eqn:lem21} and \eqref{eqn:lem22} hold for $p - 1$, and we want to show that they also hold for $ p $.
	
	First, let's prove \eqref{eqn:lem22}. There are two possible cases: (a) if $0\leq C'(p-1)<w_i$, then both $(p-1)$th and $p$th emission opportunities occur in the same round, thus $C'(p)=C'(p-1)+1$. By the induction hypothesis, $\left( C + p-2\right) \mymod  w_i + 1 < w_i$, i.e., $\left( C + p-2\right) \mymod  w_i <w_i - 1$. Note that, for any integer $x$
 	\begin{equation}\label{eq:mod}
	(x + 1)\mymod w  = \begin{cases} (x\mymod w) + 1 	\;\mif (x\mymod w) < w - 1&\\
	0			\motherwise
	\end{cases}
	\end{equation}
	By using \eqref{eq:mod}, we obtain that $C'(p)$ is given by \eqref{eqn:lem22} as required. (b) In the second case, $C'(p-1)=w_i$ then the next emission opportunity occurs in the first cycle of the next round, thus $C'(p)=1$. Here too, applying \eqref{eq:mod} shows that $C'(p)$ is given by \eqref{eqn:lem22} as required.

%
%
	 Then, we prove \eqref{eqn:lem21}. Let $N$ be the number of emission opportunities for flow $j$ in 
$[s , \tau_{\sigma(p)})$. $N$ is the sum of $N_1$, the number of emission opportunities in $[s , \tau_{\sigma(p-1)})$, and $N_2$, the number of emission opportunities in $(\tau_{\sigma(p - 1)} , \tau_{\sigma(p)})$.
By the induction hypothesis, $N_1 \leq q'_{i,j}\left(C , p - 1 \right)$. Also, by Lemma \ref{lem:NVg}, we have $N_2 \leq q_{i,j}(C'(p))$. Thus, by using \eqref{eqn:lem22} which was just shown to also hold for $p$, we obtain
	\begin{equation}	
	\begin{aligned}
	N \leq & \sum_{k=0}^{p - 1} q_{i,j} \left( \left( C + k - 1\right) \mymod   w_i + 1 \right) \\
	& + q_{i,j}\left(\left( C + p - 1\right) \mymod   w_i + 1 \right)
	\end{aligned}
	\end{equation}
	where the right-hand side is equal to $ q'_{i,j}(C,p)$ as required.
\end{proof}

\begin{lemma}\label{lem:numg}
	For any backlogged period $(s,t]$ of flow $i$ with $p$ complete services, the number of emission opportunities for flow $j \neq i$ in $(s , \tau_{\sigma(p)})$ is upper bounded by $\phi_{i,j}(p)$, 
defined in \eqref{eq:phi}.
\end{lemma}

\begin{proof}
	
	Lemma \ref{lem:NVg2} gives the number of emission opportunities for flow $j \neq i$ in $(s , \tau_{\sigma(p)})$, for any backlogged period $(s,t]$ of flow $i$ with $p$ complete services, when the first service starts at cycle number $C$ (cycle number at time $\tau_{\sigma(0)}$). To obtain the lemma, we maximize this result over $C$. We show the following properties.

	\noindent\textbf{(P1)} For any integer $C \in [1 , w_i]$,
	\begin{equation}\label{eq:temp21}
	\sum_{k=0}^{w_i - 1} q_{i,j} \left (\left( C + k - 1\right) \mymod   w_i + 1 \right) = w_j
	\end{equation}
The mapping $k\mapsto \left( C + k - 1\right) \mymod   w_i + 1$ is one-to-one from $\lc 0,...,w_i-1\rc$ onto $\lc 1,...,w_i\rc$, thus the left-hand side of \eqref{eq:temp21} is equal to $\sum_{k=1}^{w_i} q_{i,j} \left(k\right)$ that as we show now, is equal to $w_j$.
%
%
%
First, we have $q_{i,j}(1) =  \left[w_j - w_i\right]^+ + 1$. Also, $q_{i,j}(k)=1$ when $k > 1$ and $w_j \geq k+1$. Thus, $\sum_{k=2}^{w_i} q_{i,j} \left(k \right)=\min(w_i - 1, w_j - 1)$ and finally
the left-hand side is equal to $ \left[w_j - w_i\right]^+  +  \min(w_i - 1, w_j - 1) + 1$, which is equal to $w_j$.
	
	\noindent\textbf{(P2)} For any integers $C \in [1 , w_i]$ and $p \geq 0$, $q'_{i,j}\left(C , p \right)=$
	\begin{equation}\label{eq:temp3}
	 \left \lfloor \frac{p}{w_i} \right \rfloor w_j
+
	 \sum_{k=0}^{p \mymod   w_i} q_{i,j} \left( \left( C + k - 1\right) \mymod   w_i + 1 \right)
	\end{equation}
	
	$q_{i,j}$ is a periodic function with period $w_i$. By (P1), the sum over one complete period is $w_j$. Also, we can write $ p = \left \lfloor \frac{p}{w_i} \right \rfloor w_i + p~ \mymod   w_i$. Thus, we have $\left \lfloor \frac{p}{w_i} \right \rfloor$ complete rounds, and the sum in \eqref{eq:temp3} is the remainder.

	\noindent\textbf{(P3)} $q_{i,j}$ is a wide-sense decreasing function.
This means that for any integer $k \in [1 , w_i)$, $q_{i,j}(k + 1) \leq q_{i,j}(k)$. If $k = 1$, this follows from $q_{i,j}(1)\geq 1$ and $q_{i,j}(2)\leq 1$. Else if $k \leq w_j < k + 1$, then $q_{i,j}(k + 1) = 0$ and $q_{i,j}(k) = 1$. Else, they are  equal. Hence, in all cases the property holds.

	\noindent\textbf{(P4)} For any integer $C \in [1 , w_i]$ and $p \geq 0$,
	\begin{equation}\label{eq:temp5}
	\begin{aligned}
	q'_{i,j}\left(C , p \right) \leq  q'_{i,j}\left(1 , p \right)
	\end{aligned}
	\end{equation}
	By using (P2), we should show that $\sum_{k=0}^{p \mymod   w_i} q_{i,j} \left( \left( C + k - 1\right) \mymod   w_i + 1 \right)$ is upper bounded by $\sum_{k=0}^{p \mymod   w_i} q_{i,j} \left(  k \mymod   w_i + 1 \right)$. Note that here we have $k \mymod w_i = k$. Both sides are the sum of $a\isdef p~\mymod   w_i + 1$ unique elements of the set $\{q_{i,j}(k) \}_{k \in [1 , w_i]}$.  By (P3), the right-hand side is the maximum sum of $a$ unique elements of this set.
	
	\noindent\textbf{(P5)} For any integer $p \geq 0$,
	\begin{equation}\label{eq:temp4}
	\begin{aligned}
	q'_{i,j}\left(1 , p \right) &=  \phi_{i,j}(p)
	\end{aligned}
	\end{equation}
	
		We apply (P2) with $C =1$ to compute $q'_{i,j}\left(1 , p \right)$. Then, the sum in the right-hand side of \eqref{eq:temp3} is equal to $\sum_{k=0}^{p ~\mymod    w_i } q_{i,j} \left(k + 1 \right)$, as $k \mymod w_i = k$. Then, by using the same argument after \eqref{eq:temp21}, it is equal to $\left[w_j - w_i\right]^+ + 1 +  \min(p~ \mymod   w_i , w_j - 1)$, which, by \eqref{eq:phi}, is precisely $\phi_{i,j}(p)$.

The lemma then follows directly from (P4) and (P5).
%
%
\end{proof}

\begin{lemma} \label{lem:g}
	For every flow $j \neq i$,
	\begin{equation}
 R_j^*(t)  \leq R_j^*(\tau_{\sigma(p)})
	\end{equation}
\end{lemma}

\begin{proof}
If $t \leq \tau_{\sigma(p)}$, the result follows from $R_j^*$ being wide-sense increasing. Else, we have $t > \tau_{\sigma(p)}$; this implies that flow $i$ is served during $[\tau_{\sigma(p)},t]$; thus for any other flow $j$, $R_j^*(t) =R_j^*(\tau_{\sigma(p)})$.
\end{proof}

\subsection{Amount of Service to Flow of Interest}

\begin{lemma}\label{prop:p}
The number of complete services, $p$, of flow of interest, $i$, in $(s,t]$ is upper bounded by:
	\begin{equation}\label{eqn:p}
		p \leq \left \lfloor \frac{R_i^*(t) - R_i^*(s)}{\lmin_i} \right \rfloor
	\end{equation}
\end{lemma}

\begin{proof}	
First, $R_i^*(s) \leq R_i^*(\tau_{\sigma(0)})$, as $s \leq \tau_{\sigma(0)}$ and $R_i^*$ is wide-sense increasing. Second, consider the two cases in \ref{ssec:t}. If $t \geq \tau_{\sigma(p)}$, the property holds. Else, the scheduler in not serving flow $i$ in $[ \tau_{\sigma(p - 1) + 1} ,  \tau_{\sigma(p)})$, thus, $R_i^*(t) = R_i^*(\tau_{\sigma(p)})$. Hence, in both cases $R_i^*(t) \geq R_i^*(\tau_{\sigma(p)})$. By \eqref{eqn:Min},
$
	R_i^*(t) - R_i^*(s)  \geq p\lmin_i
$. Then, observe that $p$ is integer.
\end{proof}

\subsection{Total Amount of  Service}
\begin{lemma} \label{lem:TotService}
For any backlogged period $(s , t]$ of the flow of interest $i$,
	\begin{equation}\label{eqn:TotService}
		\beta(t-s) \leq \psi_i \left( R_i^*(t) - R_i^*(s) \right)
	\end{equation}
where $\psi_i$ is defined in \eqref{eq:psi}.
\end{lemma}
\begin{proof}
As the interval $(s,t]$ is a backlogged period, by the definition of the strict service curve for the aggregate of flows, $\beta(t-s) \leq\sum_{j} R_j^*(t) - R_j^*(s)$.
We upper bound $ R_j^*(t)$ for all $j \neq i$ by applying Lemma \ref{lem:g},
	\begin{equation}
		\beta(t-s) \leq (R_i^*(t) - R_i^*(s)) + \sum_{j,j \neq i} R_j^*(\tau_{\sigma(p)}) - R_j^*(s)
	\end{equation}
Each flow $j$ has at most $\phi_{i,j}(p)$ emission opportunities during $\left(s,\tau_{\sigma(p)}\right)$ (Lemma~\ref{lem:numg}) and can send at most
one packet of maximum size in each. Thus,
	\begin{equation}
		\beta(t-s) \leq (R_i^*(t) - R_i^*(s)) + \sum_{j,j \neq i} \phi_{i,j}(p)\lmax_j
	\end{equation}
Also, Lemma \ref{prop:p} finds an upper bound on $p$. Thereby,
	\begin{equation}
		\begin{aligned}
		\beta(t-s) & \leq (R_i^*(t) - R_i^*(s)) \\
				&+\sum_{j,j \neq i} \phi_{i,j}\left( \left \lfloor \frac{R_i^*(t) - R_i^*(s)}{\lmin_i } \right \rfloor \right)\lmax_j
		\end{aligned}
	\end{equation}
	where the right-hand side is equal to $\psi_i(R_i^*(t) - R_i^*(s))$.
\end{proof}
\subsection{Lower Pseudo-inverse of $\psi_i$}
Our next step is to invert \eqref{eqn:TotService} by computing the lower-pseudo inverse of $\psi_i$. As the calculus of pseudo inverses applies to wide-sense increasing functions, we first show:
\begin{lemma} \label{prop:wsi}
	$\psi_i$, defined in \eqref{eq:psi}, is wide-sense increasing.
\end{lemma}

\begin{proof}
	It is sufficient to show that $\phi_{i,j}$, defined in \eqref{eq:phi}, is a wide-sense increasing function. For any non-negative integers $x$ and $y$ such that $y \leq x$, we can write $x = kw_i + (x~\mymod   w_i)$ and $y = k'w_i+ (y~\mymod   w_i)$, where $k,k'$ are non-negative integers. We must have $k\leq k'$. If $k=k'$, we know that $(y~ \mymod   w_i \leq x~  \mymod   w_i)$ and $\left \lfloor \frac{x}{w_i} \right \rfloor =\left \lfloor \frac{y}{w_i} \right \rfloor$. Hence, $ \phi_{i,j}(y) \leq \phi_{i,j}(x)$. Else, $k > k'$ and $\left \lfloor \frac{x}{w_i} \right \rfloor > \left \lfloor \frac{y}{w_i} \right \rfloor$. Thereby, $\phi_{i,j}(x)$ is at least one $w_j$ larger than $\phi_{i,j}(y)$. Hence, $ \phi_{i,j}(y) < \phi_{i,j}(x)$.
\end{proof}

\begin{lemma}\label{lem:inverse1}
	Let $g_0, g_1, \ldots, g_k,\ldots$ be a non-negative sequence such that $g_{k+1}-g_k\geq 1$. The sequence can be extended to a function in $\mathscr{F}$ by $g(x)=g_{\lfloor x\rfloor}$ and let $g^{\downarrow}$ be its lower pseudo-inverse, so that $g^{\downarrow}(y)=k  +1\in \Nats\Leftrightarrow g_k < y \leq g_{k+1}$. Define $f\in \mathscr{F}$ by $f(x)=g_{\lfloor x\rfloor}+x\mymod 1$. Then, $f^{\downarrow}=\lambda_1\otimes g^{\downarrow}$.
\end{lemma}
\begin{proof}
Observe that convolving $g^{\downarrow}$ with $\lambda_1$ consists in smoothing the unit steps with a slope of $1$ (Fig.~\ref{fig:minplus}). Thus $(\lambda_1\otimes g^{\downarrow})(y)=k +y-g_k$ whenever $g_k\leq y\leq g_{k}+1$ and $(\lambda_1\otimes g^{\downarrow})(y)=k  +1$ whenever $g_k+1\leq y\leq g_{k+1}$.

Also, $f$ is piecewise linear and can be inverted in closed form on every interval where it is linear. A direct calculation gives $f^{\downarrow}(y)=k+y-g_k$ whenever $g_k\leq y\leq g_{k}+1$ and $f^{\downarrow}(y)=k + 1$ whenever $g_k+1\leq y\leq g_{k+1}$.
\end{proof}

\begin{lemma}\label{lem:inverse2}
	Let $f\in\mathscr{F}$ and $l,m>0$. Define $h\in\mathscr{F}$ by $h(x)=mf\lp\frac{x}{l}\rp$. Then, for all $y\geq 0$, $h^{\downarrow}(y)=lf^{\downarrow}\lp\frac{y}{m}\rp$.
\end{lemma}
\begin{proof}
Let $B(f,y)\isdef \lc x\geq 0, h(x)\geq y \rc$ so that $f^{\downarrow}(y)=\inf B(y,f)$. Observe that $x\in B(h,y)\Leftrightarrow
\frac{x}{l}\in B\lp f, \frac{y}{m}\rp$.
\end{proof}

\begin{lemma}\label{lem:inverse3}
	Let $a\in\mathscr{F}$ and $l>0$. Define $b\in\mathscr{F}$ by $b(x)=lf\lp\frac{x}{l}\rp$. Then, for all $x\geq 0$, $(\lambda_1\otimes b)(x)=l (\lambda_1\otimes a)\lp\frac{x}{l}\rp$.
\end{lemma}
\begin{proof}
Do the change of variable $u=lv$ in the expansion $(\lambda_1\otimes b)(x)=
\inf_{0\leq u\leq x}\lp u+b(x-u)\rp$ and obtain $(\lambda_1\otimes b)(x)=
\inf_{0\leq v\leq \frac{x}{l}}\lp lv + a\lp\frac{x}{l}-v\rp\rp=l\lp \lambda_1\otimes a\rp\lp\frac{x}{l}\rp$.  \end{proof}

We can now compute the lower-pseudo inverse of $\psi_i$. First, define the sequence $g$ by $g_k=\frac{1}{\lmin_i}\psi_i\lp k \lmin_i\rp$. As in Lemma~\ref{lem:inverse1}, $g$ can be extended to a piecewise constant function whose lower-pseudo inverse, $g^{\downarrow}$, can be directly computed:
\begin{equation}
g^{\downarrow}(x)=\frac{1}{\lmin_i}\sum_{k=0}^{w_i - 1}\nu_{\lmin_i,L_{\tot}}\left(\lmin_i\lb  x - g_k\rb^+\right)
\label{eq:ginv}
\end{equation}

Second, observe that for all $x \geq 0$, $\psi_i(x) = \psi_i(\lfloor \frac{x}{\lmin_i} \rfloor \lmin_i) + x \mymod \lmin_i$. Define $f$ and $h$ from $g$ as in Lemmas~\ref{lem:inverse1} and~\ref{lem:inverse2} with $l=m=\lmin_i$, so that $h=\psi_i$. Apply  Lemmas~\ref{lem:inverse1} and~\ref{lem:inverse2}
and obtain
$
\psi_i^{\downarrow}(x)=\lmin_i  \lp\lambda_1\otimes g^{\downarrow}\rp(\frac{x}{\lmin_i})
$. 
Now apply Lemma~\ref{lem:inverse3} with $a=g^{\downarrow}$, $l=\lmin_i$, and $b=U_i$ to obtain

\begin{equation}\label{eq:psiinv}
   \psi_i^{\downarrow} =  \lambda_1 \otimes U_i
\end{equation}

\subsection{Proof of Theorem \ref{Thm:1}}
\begin{proof}
Lemma~\ref{lem:TotService} gives, in \eqref{eqn:TotService}, an upper bound on the total amount of service as a function of the service received by the flow of interest.
We invert \eqref{eqn:TotService} by the lower-pseudo inverse technique in \eqref{lem:lsi} and obtain
$		
R_i^*(t) - R_i^*(s) \geq \psi_i^{\downarrow}(\beta(t-s))
$.
The lower-pseudo inverse of $\psi_i$ is given by \eqref{eq:psiinv}, thus
	\begin{equation}
	R_i^*(t) - R_i^*(s)  \geq  \left( \lambda_1 \otimes U_i \right) \left( \beta \left( t - s \right) \right)=\beta_i \left( t - s\right)
	\end{equation}

Lastly, we need to prove that $\beta_i$ is super-additive. This follows from the tightness result in \thref{thm:tight} (the proof of which is independent of rest of this proof). Indeed, the super-additive closure $\bar{\beta}_i$ of $\beta_i$ is also a strict service curve, and $\bar{\beta}_i(t)\geq \beta_f(t)$ for all $t$ \cite[Prop. 5.6]{bouillard_deterministic_2018}). By \thref{thm:tight}, we also have $\bar{\beta}_i(t)\leq \beta_i(t)$ for all $t$, hence $\bar{\beta}_i=\beta_i$.	
\end{proof}

\section{Tightness} \label{sec:tightness}
We first show that the strict service curve we have obtained is the best possible. Proofs of results in this Section are in \iftr Appendix\else\cite{tabatabaee2020interleaved}\fi.
\subsection{Tightness of Strict Service Curve}
\begin{theorem}[Tightness of the IWRR Service Curve]\label{thm:tight}
Consider a weighted round-robin subsystem that uses the IWRR scheduling algorithm, as defined in \sref{sec:alg}. Assume the following system parameters are fixed: the number of input flows, the weight $w_j$ allocated to every flow $j$, the bounds on packet sizes $\lmin_j$ and $\lmax_j$ for every flow $j$, and the strict service curve $\beta$ for the aggregate of all flows. Let $i$ be the index of one of the flows.

Assume that $b_i\in\mathscr{F}$ is a strict service curve for flow $i$ in any system that satisfies the specifications above. Then $b_i\leq \beta_i$ where $\beta_i$ is given in \thref{Thm:1}.
%
%
%
%
\end{theorem}
Interestingly, we obtain a similar result for WRR. Recall that $\beta_i'$ is the strict service curve for flow $i$, described in \eqref{eq:WRR}, which was obtained in \cite[Sec.~8.2.4]{bouillard_deterministic_2018}.
\begin{theorem}[Tightness of the WRR Service Curve]\label{thm:tightWRR}
	Theorem \ref{thm:tight} is also valid if we replace IWRR with WRR. Specifically,  using WRR as a scheduling policy, $\beta_i'$ is the largest possible strict service curve for flow $i$.
\end{theorem}

\subsection{Tightness of Delay Bounds With Constant Packet Sizes}
Having obtained the best-possible strict service curve does not guarantee that the delay bounds derived from it are tight, i.e., are worst-case delays. This is because a service curve is only an abstraction of the system; and we have obtained a strict service curve, and non-strict service curves might provide better results. However, we show that, for flows of packets of constant size, 
we do obtain tight delay bounds. We show that it holds for IWRR and for WRR.

Recall that a delay bound requires the knowledge of an arrival curve $\alpha_i$ for the flow of interest. If this flow generates only packets of length $l$, then $\alpha_i$ can be assumed to be a multiple of $l$ and sub-additive. A delay bound for this flow is then equal to $h(\alpha_i,\beta_i)$ (see \eqref{eq:horiz}).

\begin{theorem}[Tightness of Delay Bound for IWRR with Constant Packet Size] \label{thm:dbIWRR}
Consider a system, as in \thref{thm:tight}, with the additional assumption that, for the flow of interest $i$, $\lmin_i =\lmax_i=l$.

Let $ \alpha_i\in\mathscr{F} $ be a sub-additive function that is an integer multiple of $l$, and assume that flow $i$ has $ \alpha_i$ as arrival curve. The network calculus delay bound is tight, i.e, there exists a trajectory where the delay of one packet of flow $i$ is equal to $h(\alpha_i,\beta_i)$.
\end{theorem}

\begin{theorem}[Tightness of Delay Bound for WRR with Constant Packet Size]\label{thm:dbWRR}
	Theorem \ref{thm:dbIWRR} is also valid for the WRR policy.
\end{theorem}

    \section{Numerical Examples}
\label{sec:numerics}
To compare IWRR and WRR worst-case delays, we provide some numerical examples. First, we consider a system of 8 input flows $f_1,\ldots,f_8$ with respective weights $\{ 22,27,28,30,30,34,41,45\} $ and $\lmin= \lmax= l =7119\text{ bit}$. Let the aggregate service, $\beta$, be a constant bit rate of 10 Mb/s.
%
For every flow $i$, we compute the IWRR and WRR strict service curves $\beta_i, \beta'_i$.
Then, for every $i$, we generate $N=1000$
 leaky-bucket arrival curves $\gamma_{r,b_k}$,  $k=1\ldots N$,
with rate $r=0.5$~Mb/s and burst $b_k$ picked uniformly at random in $[1 , 20]$
packets. Then, we use $\alpha^k_i = \lceil \frac{\gamma_{r,b_k}}{l}
\rceil l$ to satisfy the conditions of Theorems \ref{thm:dbIWRR} and \ref{thm:dbWRR}
and to compute $d^k_i=h(\alpha^k_i,\beta_i)$ and $\dot{d}^k_i=h(\alpha^k_i,\beta'_i)$.
Fig.~\ref{fig:exp21} gives the box-and-whisker plots of the $\dot{d}^k_i-d^k_i$ series. The median of WRR delay bounds $\dot{d}^k_i$ are also provided to illustrate the improvement.

Second, we repeated the same study for $ M = 10000$ sets of system parameters.
For each system, we choose the weights of 8 flows by picking them uniformly at random between 10 and 50, and we pick a packet length $l$ uniformly at random between 64 to 1522 bytes. For each experiment, we call flow 1 the flow with the smallest weight, flow 2 with second smallest weight, and so on. As the scale of delay bounds depends on the choices of weights and the packet length, the $\dot{d}^k_i-d^k_i$ series are divided by $\dot{d}^{\bar{m}}_i$, the median of WRR delay bounds for flow $i$. Fig.~\ref{fig:exp21} gives the box-and-whisker plots of the $\frac{\dot{d}^k_i-d^k_i}{\dot{d}^{\bar{m}}_i}$ series. Using IWRR improves worst-case delays, as expected, and the improvement is larger for flows with larger weights.

	\begin{figure} [htbp]
        \scalebox{0.5}{
%
%
\definecolor{mycolor1}{rgb}{0.00000,0.44700,0.74100}%
\definecolor{mycolor2}{rgb}{0.85000,0.32500,0.09800}%
\begin{tikzpicture}

\begin{axis}[%
width=6.028in,
height=4.754in,
at={(0in,0in)},
scale only axis,
unbounded coords=jump,
xmin=0.5,
xmax=8.5,
xtick={1,2,3,4,5,6,7,8},
ymin=-5,
ymax=180.1274,
title ={\Large Absolute Improvement of Delay Bounds of IWRR wrt WRR on one Configuration},
ylabel = {\large$\text{WRR Worst-case Delay} -  \text{IWRR Worst-case Delay} $ (Time ($ms$))},
xlabel = {Flows},
axis background/.style={fill=white},
legend style={at={(axis cs:5,20)},anchor=north west , legend cell align=left, align=left, draw=white!1!black}
]

\draw [-> , thick, color = mycolor1] (3,160) -- (3.1,165)node[pos=0.5, right] {Median of WRR Worst-case Delay};

table[row sep=crcr]{%
	1	173.7036\\
	2	170.1441\\
	3	169.4322\\
	4	168.0084\\
	5	168.0084\\
	6	165.1608\\
	7	160.1775\\
	8	157.3299\\
};

\addplot [loosely dashed, color=mycolor1, mark=diamond*, mark options={solid, mycolor1} ,mark size=3pt]
table[row sep=crcr]{%
	1	173.7036\\
	2	170.1441\\
	3	169.4322\\
	4	168.0084\\
	5	168.0084\\
	6	165.1608\\
	7	160.1775\\
	8	157.3299\\
};
\addlegendentry{Median of WRR Worst-case Delay};

table[row sep=crcr]{%
	1	104.6493\\
	2	126.0063\\
	3	129.5658\\
	4	135.261\\
	5	135.261\\
	6	140.9562\\
	7	145.9395\\
	8	145.9395\\
};


\addplot [color=black, dashed, forget plot]
table[row sep=crcr]{%
	1	84.7161\\
	1	104.6493\\
};
\addplot [color=black, dashed, forget plot]
table[row sep=crcr]{%
	2	106.0731\\
	2	126.0063\\
};
\addplot [color=black, dashed, forget plot]
table[row sep=crcr]{%
	3	109.6326\\
	3	129.5658\\
};
\addplot [color=black, dashed, forget plot]
table[row sep=crcr]{%
	4	115.3278\\
	4	135.261\\
};
\addplot [color=black, dashed, forget plot]
table[row sep=crcr]{%
	5	115.3278\\
	5	135.261\\
};
\addplot [color=black, dashed, forget plot]
table[row sep=crcr]{%
	6	121.023\\
	6	140.9562\\
};
\addplot [color=black, dashed, forget plot]
table[row sep=crcr]{%
	7	126.0063\\
	7	145.9395\\
};
\addplot [color=black, dashed, forget plot]
table[row sep=crcr]{%
	8	126.0063\\
	8	145.9395\\
};
\addplot [color=black, dashed, forget plot]
table[row sep=crcr]{%
	1	9.9666\\
	1	34.8831\\
};
\addplot [color=black, dashed, forget plot]
table[row sep=crcr]{%
	2	31.3236\\
	2	56.2401\\
};
\addplot [color=black, dashed, forget plot]
table[row sep=crcr]{%
	3	34.8831\\
	3	59.7996\\
};
\addplot [color=black, dashed, forget plot]
table[row sep=crcr]{%
	4	40.5783\\
	4	65.4948\\
};
\addplot [color=black, dashed, forget plot]
table[row sep=crcr]{%
	5	40.5783\\
	5	65.4948\\
};
\addplot [color=black, dashed, forget plot]
table[row sep=crcr]{%
	6	46.2735\\
	6	71.19\\
};
\addplot [color=black, dashed, forget plot]
table[row sep=crcr]{%
	7	51.2568\\
	7	76.1733\\
};
\addplot [color=black, dashed, forget plot]
table[row sep=crcr]{%
	8	51.2568\\
	8	76.1733\\
};
\addplot [color=black, forget plot]
table[row sep=crcr]{%
	0.875	104.6493\\
	1.125	104.6493\\
};
\addplot [color=black, forget plot]
table[row sep=crcr]{%
	1.875	126.0063\\
	2.125	126.0063\\
};
\addplot [color=black, forget plot]
table[row sep=crcr]{%
	2.875	129.5658\\
	3.125	129.5658\\
};
\addplot [color=black, forget plot]
table[row sep=crcr]{%
	3.875	135.261\\
	4.125	135.261\\
};
\addplot [color=black, forget plot]
table[row sep=crcr]{%
	4.875	135.261\\
	5.125	135.261\\
};
\addplot [color=black, forget plot]
table[row sep=crcr]{%
	5.875	140.9562\\
	6.125	140.9562\\
};
\addplot [color=black, forget plot]
table[row sep=crcr]{%
	6.875	145.9395\\
	7.125	145.9395\\
};
\addplot [color=black, forget plot]
table[row sep=crcr]{%
	7.875	145.9395\\
	8.125	145.9395\\
};
\addplot [color=black, forget plot]
table[row sep=crcr]{%
	0.875	9.9666\\
	1.125	9.9666\\
};
\addplot [color=black, forget plot]
table[row sep=crcr]{%
	1.875	31.3236\\
	2.125	31.3236\\
};
\addplot [color=black, forget plot]
table[row sep=crcr]{%
	2.875	34.8831\\
	3.125	34.8831\\
};
\addplot [color=black, forget plot]
table[row sep=crcr]{%
	3.875	40.5783\\
	4.125	40.5783\\
};
\addplot [color=black, forget plot]
table[row sep=crcr]{%
	4.875	40.5783\\
	5.125	40.5783\\
};
\addplot [color=black, forget plot]
table[row sep=crcr]{%
	5.875	46.2735\\
	6.125	46.2735\\
};
\addplot [color=black, forget plot]
table[row sep=crcr]{%
	6.875	51.2568\\
	7.125	51.2568\\
};
\addplot [color=black, forget plot]
table[row sep=crcr]{%
	7.875	51.2568\\
	8.125	51.2568\\
};
\addplot [color=blue, forget plot]
table[row sep=crcr]{%
	0.75	34.8831\\
	0.75	84.7161\\
	1.25	84.7161\\
	1.25	34.8831\\
	0.75	34.8831\\
};
\addplot [color=blue, forget plot]
table[row sep=crcr]{%
	1.75	56.2401\\
	1.75	106.0731\\
	2.25	106.0731\\
	2.25	56.2401\\
	1.75	56.2401\\
};
\addplot [color=blue, forget plot]
table[row sep=crcr]{%
	2.75	59.7996\\
	2.75	109.6326\\
	3.25	109.6326\\
	3.25	59.7996\\
	2.75	59.7996\\
};
\addplot [color=blue, forget plot]
table[row sep=crcr]{%
	3.75	65.4948\\
	3.75	115.3278\\
	4.25	115.3278\\
	4.25	65.4948\\
	3.75	65.4948\\
};
\addplot [color=blue, forget plot]
table[row sep=crcr]{%
	4.75	65.4948\\
	4.75	115.3278\\
	5.25	115.3278\\
	5.25	65.4948\\
	4.75	65.4948\\
};
\addplot [color=blue, forget plot]
table[row sep=crcr]{%
	5.75	71.19\\
	5.75	121.023\\
	6.25	121.023\\
	6.25	71.19\\
	5.75	71.19\\
};
\addplot [color=blue, forget plot]
table[row sep=crcr]{%
	6.75	76.1733\\
	6.75	126.0063\\
	7.25	126.0063\\
	7.25	76.1733\\
	6.75	76.1733\\
};
\addplot [color=blue, forget plot]
table[row sep=crcr]{%
	7.75	76.1733\\
	7.75	126.0063\\
	8.25	126.0063\\
	8.25	76.1733\\
	7.75	76.1733\\
};
\addplot [color=red, forget plot]
table[row sep=crcr]{%
	0.75	59.7996\\
	1.25	59.7996\\
};
\addplot [color=red, forget plot]
table[row sep=crcr]{%
	1.75	81.1566\\
	2.25	81.1566\\
};
\addplot [color=red, forget plot]
table[row sep=crcr]{%
	2.75	84.7161\\
	3.25	84.7161\\
};
\addplot [color=red, forget plot]
table[row sep=crcr]{%
	3.75	90.4113\\
	4.25	90.4113\\
};
\addplot [color=red, forget plot]
table[row sep=crcr]{%
	4.75	90.4113\\
	5.25	90.4113\\
};
\addplot [color=red, forget plot]
table[row sep=crcr]{%
	5.75	96.1065\\
	6.25	96.1065\\
};
\addplot [color=red, forget plot]
table[row sep=crcr]{%
	6.75	101.0898\\
	7.25	101.0898\\
};
\addplot [color=red, forget plot]
table[row sep=crcr]{%
	7.75	101.0898\\
	8.25	101.0898\\
};

\end{axis}

\begin{axis}[%
width=7.778in,
height=4.754in,
at={(0in,0in)},
scale only axis,
xmin=0,
xmax=1,
ymin=0,
ymax=1,
axis line style={draw=none},
ticks=none,
axis x line*=bottom,
axis y line*=left,
legend style={legend cell align=left, align=left, draw=white!15!black}
]

\end{axis}
\end{tikzpicture}%
}
	\caption{\sffamily \small Box-and-whisker plots of difference between WRR and IWRR delay bounds with weights  $\{ 22,27,28,30,30,34,41,45\} $ and $l = 7119 $ bit with random arrival curves. Median WRR delay bounds are also provided. }
	\label{fig:exp21}
\end{figure}
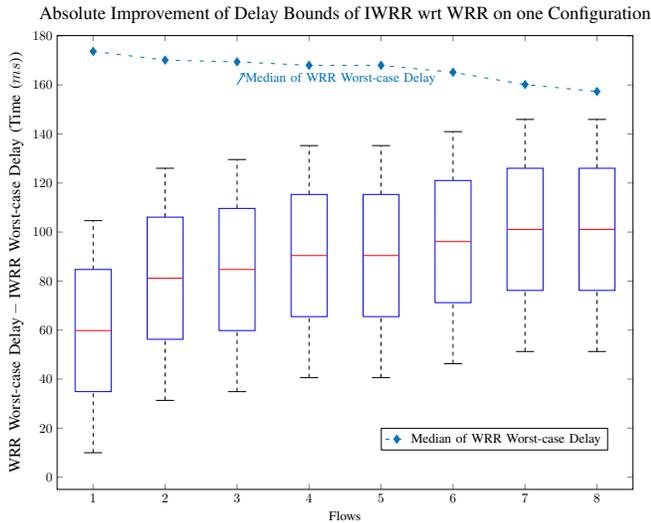
	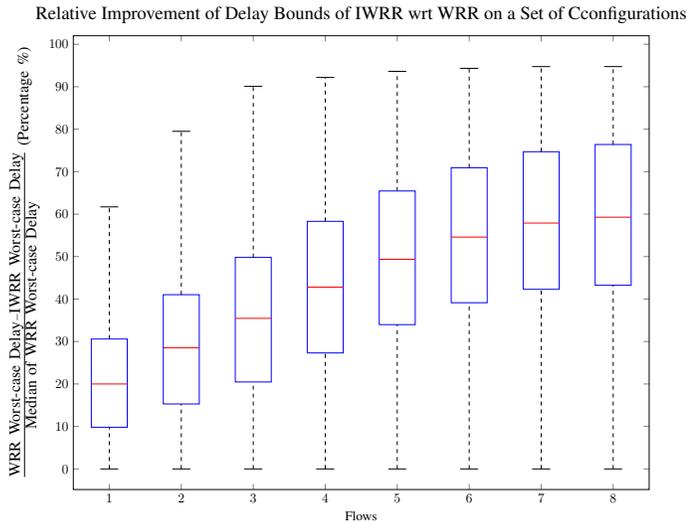
\begin{figure} [htbp]
          \scalebox{0.5}{
%
%
\begin{tikzpicture}

\begin{axis}[%
width=6.028in,
height=4.754in,
at={(0in,0in)},
scale only axis,
xmin=0.5,
xmax=8.5,
xtick={1,2,3,4,5,6,7,8},
ymin=-4.85701193387848,
ymax=101.997250611448,
ylabel = {\large$\frac{\text{WRR Worst-case Delay} -  \text{IWRR Worst-case Delay} }{\text{Median of WRR Worst-case Delay} }$ (Percentage \%)},
xlabel = {Flows},
title ={\Large Relative Improvement of Delay Bounds of IWRR wrt WRR on a Set of Cconfigurations},
axis background/.style={fill=white},
legend style={legend cell align=left, align=left, draw=white!15!black}
]
\addplot [color=black, dashed, forget plot]
table[row sep=crcr]{%
	1	30.5825242718447\\
	1	61.7021276595745\\
};
\addplot [color=black, dashed, forget plot]
table[row sep=crcr]{%
	2	41.025641025641\\
	2	79.4964028776978\\
};
\addplot [color=black, dashed, forget plot]
table[row sep=crcr]{%
	3	49.802371541502\\
	3	90.0793650793651\\
};
\addplot [color=black, dashed, forget plot]
table[row sep=crcr]{%
	4	58.2995951417004\\
	4	92.1824104234528\\
};
\addplot [color=black, dashed, forget plot]
table[row sep=crcr]{%
	5	65.4450261780105\\
	5	93.6170212765958\\
};
\addplot [color=black, dashed, forget plot]
table[row sep=crcr]{%
	6	70.9183673469388\\
	6	94.3262411347518\\
};
\addplot [color=black, dashed, forget plot]
table[row sep=crcr]{%
	7	74.6887966804979\\
	7	94.7368421052632\\
};
\addplot [color=black, dashed, forget plot]
table[row sep=crcr]{%
	8	76.3888888888889\\
	8	94.7368421052632\\
};
\addplot [color=black, dashed, forget plot]
table[row sep=crcr]{%
	1	0\\
	1	9.81308411214954\\
};
\addplot [color=black, dashed, forget plot]
table[row sep=crcr]{%
	2	0\\
	2	15.2892561983471\\
};
\addplot [color=black, dashed, forget plot]
table[row sep=crcr]{%
	3	0\\
	3	20.4724409448819\\
};
\addplot [color=black, dashed, forget plot]
table[row sep=crcr]{%
	4	0\\
	4	27.3076923076923\\
};
\addplot [color=black, dashed, forget plot]
table[row sep=crcr]{%
	5	0\\
	5	33.9712918660287\\
};
\addplot [color=black, dashed, forget plot]
table[row sep=crcr]{%
	6	0\\
	6	39.1304347826087\\
};
\addplot [color=black, dashed, forget plot]
table[row sep=crcr]{%
	7	0\\
	7	42.2885572139304\\
};
\addplot [color=black, dashed, forget plot]
table[row sep=crcr]{%
	8	0\\
	8	43.2432432432432\\
};
\addplot [color=black, forget plot]
table[row sep=crcr]{%
	0.875	61.7021276595745\\
	1.125	61.7021276595745\\
};
\addplot [color=black, forget plot]
table[row sep=crcr]{%
	1.875	79.4964028776978\\
	2.125	79.4964028776978\\
};
\addplot [color=black, forget plot]
table[row sep=crcr]{%
	2.875	90.0793650793651\\
	3.125	90.0793650793651\\
};
\addplot [color=black, forget plot]
table[row sep=crcr]{%
	3.875	92.1824104234528\\
	4.125	92.1824104234528\\
};
\addplot [color=black, forget plot]
table[row sep=crcr]{%
	4.875	93.6170212765958\\
	5.125	93.6170212765958\\
};
\addplot [color=black, forget plot]
table[row sep=crcr]{%
	5.875	94.3262411347518\\
	6.125	94.3262411347518\\
};
\addplot [color=black, forget plot]
table[row sep=crcr]{%
	6.875	94.7368421052632\\
	7.125	94.7368421052632\\
};
\addplot [color=black, forget plot]
table[row sep=crcr]{%
	7.875	94.7368421052632\\
	8.125	94.7368421052632\\
};
\addplot [color=black, forget plot]
table[row sep=crcr]{%
	0.875	0\\
	1.125	0\\
};
\addplot [color=black, forget plot]
table[row sep=crcr]{%
	1.875	0\\
	2.125	0\\
};
\addplot [color=black, forget plot]
table[row sep=crcr]{%
	2.875	0\\
	3.125	0\\
};
\addplot [color=black, forget plot]
table[row sep=crcr]{%
	3.875	0\\
	4.125	0\\
};
\addplot [color=black, forget plot]
table[row sep=crcr]{%
	4.875	0\\
	5.125	0\\
};
\addplot [color=black, forget plot]
table[row sep=crcr]{%
	5.875	0\\
	6.125	0\\
};
\addplot [color=black, forget plot]
table[row sep=crcr]{%
	6.875	0\\
	7.125	0\\
};
\addplot [color=black, forget plot]
table[row sep=crcr]{%
	7.875	0\\
	8.125	0\\
};
\addplot [color=blue, forget plot]
table[row sep=crcr]{%
	0.75	9.81308411214954\\
	0.75	30.5825242718447\\
	1.25	30.5825242718447\\
	1.25	9.81308411214954\\
	0.75	9.81308411214954\\
};
\addplot [color=blue, forget plot]
table[row sep=crcr]{%
	1.75	15.2892561983471\\
	1.75	41.025641025641\\
	2.25	41.025641025641\\
	2.25	15.2892561983471\\
	1.75	15.2892561983471\\
};
\addplot [color=blue, forget plot]
table[row sep=crcr]{%
	2.75	20.4724409448819\\
	2.75	49.802371541502\\
	3.25	49.802371541502\\
	3.25	20.4724409448819\\
	2.75	20.4724409448819\\
};
\addplot [color=blue, forget plot]
table[row sep=crcr]{%
	3.75	27.3076923076923\\
	3.75	58.2995951417004\\
	4.25	58.2995951417004\\
	4.25	27.3076923076923\\
	3.75	27.3076923076923\\
};
\addplot [color=blue, forget plot]
table[row sep=crcr]{%
	4.75	33.9712918660287\\
	4.75	65.4450261780105\\
	5.25	65.4450261780105\\
	5.25	33.9712918660287\\
	4.75	33.9712918660287\\
};
\addplot [color=blue, forget plot]
table[row sep=crcr]{%
	5.75	39.1304347826087\\
	5.75	70.9183673469388\\
	6.25	70.9183673469388\\
	6.25	39.1304347826087\\
	5.75	39.1304347826087\\
};
\addplot [color=blue, forget plot]
table[row sep=crcr]{%
	6.75	42.2885572139304\\
	6.75	74.6887966804979\\
	7.25	74.6887966804979\\
	7.25	42.2885572139304\\
	6.75	42.2885572139304\\
};
\addplot [color=blue, forget plot]
table[row sep=crcr]{%
	7.75	43.2432432432432\\
	7.75	76.3888888888889\\
	8.25	76.3888888888889\\
	8.25	43.2432432432432\\
	7.75	43.2432432432432\\
};
\addplot [color=red, forget plot]
table[row sep=crcr]{%
	0.75	20\\
	1.25	20\\
};
\addplot [color=red, forget plot]
table[row sep=crcr]{%
	1.75	28.5106382978723\\
	2.25	28.5106382978723\\
};
\addplot [color=red, forget plot]
table[row sep=crcr]{%
	2.75	35.4570637119113\\
	3.25	35.4570637119113\\
};
\addplot [color=red, forget plot]
table[row sep=crcr]{%
	3.75	42.8057553956834\\
	4.25	42.8057553956834\\
};
\addplot [color=red, forget plot]
table[row sep=crcr]{%
	4.75	49.3617021276596\\
	5.25	49.3617021276596\\
};
\addplot [color=red, forget plot]
table[row sep=crcr]{%
	5.75	54.5918367346939\\
	6.25	54.5918367346939\\
};
\addplot [color=red, forget plot]
table[row sep=crcr]{%
	6.75	57.8947368421053\\
	7.25	57.8947368421053\\
};
\addplot [color=red, forget plot]
table[row sep=crcr]{%
	7.75	59.2741935483871\\
	8.25	59.2741935483871\\
};
\end{axis}

\begin{axis}[%
width=7.778in,
height=4.754in,
at={(0in,0in)},
scale only axis,
xmin=0,
xmax=1,
ymin=0,
ymax=1,
axis line style={draw=none},
ticks=none,
axis x line*=bottom,
axis y line*=left,
legend style={legend cell align=left, align=left, draw=white!15!black}
]
\end{axis}
\end{tikzpicture}%
}
		\caption{\sffamily \small  Box-and-whisker plots of difference between WRR and IWRR delay bounds normalized to the median of WRR delay bounds, for several systems with weights picked uniformly at random in $[10,50]$, assigned to flow by increasing order, and a packet length picked uniformly at random in $[64 , 1522]$ bytes. 
}
		\label{fig:exp2}
	\end{figure}

	\section{Conclusion}\label{sec:conclusion}
%
%

IWRR is a variant of WRR with the same long-term rate and the same complexity. We have provided a residual strict service curve for IWRR and have showed that it is the best possible one under general assumptions. For flows with packets of constant size, we have showed that the delay bounds derived from it are worst-case. We have proved that IWRR worst-case delay is not greater than WRR and shown on experiments that the gain is significant (20\%-60\%) in practice, which speaks in favour of using IWRR as a replacement to WRR. Our result assumes that the aggregate of all IWRR queues receives a strict service curve guarantee, and we find a strict service curve guarantee for every IWRR queue. Therefore, our results apply to hierarchical schedulers. 
In future research, we plan to improve the results with supplementary hypotheses on flows, considering arrival curves and packet size distribution, with ``packet curves" \cite{bouillard2012packetization}.

	\bibliographystyle{IEEEtran}
	\vspace{-0.05in}
	\bibliography{ref,leb,boyer}
    \iftr
    \appendix

\subsection{Proof of Theorem~\ref{thm:WRRvsIWRR}  } \label{sec:sproof}
The WRR strict service curve \cite[Sec. 8.2.4]{bouillard_deterministic_2018} is defined by $\beta'_i(t)=\gamma'_i(\beta(t))$ with
\begin{align}
\label{eq:gamma'}
\gamma'_i  &=  (\lambda_1 \otimes \nu_{q_i,L_{\tot}})  \left(\lb t - Q_i \rb ^+\right)
\\
\label{eq:psi'}
\psi'_i(x) &\isdef x + \sum_{j,j \neq i} \phi'_{i,j}\left(\left\lfloor \frac{x}{\lmin_i} \right\rfloor\right)\lmax_j
\\
\label{eq:phi'}
\phi'_{i,j}(x) &\isdef \left(1 +  \left\lfloor \frac{x}{w_i} \right\rfloor \right)  w_j
\end{align}
$\gamma'_i $ is the lower-pseudo inverse of $\psi'_i$. We know that for IWRR, $\gamma_i $ is also the lower-pseudo inverse of $\psi_i$ (defined in \eqref{eq:psi}). We first show that $\psi_i \leq \psi_i'$.

It is sufficient to prove that for all $j \neq i$ and for all $k \in \mathbb{N}$, $ \phi_{i,j}(k) \leq \phi_{i,j}'(k)$. 
 From the definition of $\phi_{i,j}$ and as $\min(x \mymod w_i + 1,w_j) \leq \min(w_i,w_j)$, 
	\begin{equation}	
		\phi_{i,j}(x) \leq \left \lfloor \frac{x}{w_i} \right \rfloor w_j + \left[w_j - w_i\right]^+ + \min(w_i ,w_j)
	\end{equation}
Observe that $\left[w_j- w_i\right]^+ + \min(w_i ,w_j) = w_j$. Hence, the right-hand side is $ \phi_{i,j}'(x)$. This shows that 
\begin{equation}\label{eq:WRRless1}
 \psi_i \leq \psi_i'
\end{equation}
In \cite[Sec. 10.1]{liebeherr2017duality}, it is shown that 
\begin{equation} \label{lem:lsi2}
\forall f, g \in \mathscr{F}, f \geq g \Rightarrow f^{\downarrow} \leq g^{\downarrow}
\end{equation}

Apply \eqref{lem:lsi2} to \eqref{eq:WRRless1} to conclude the proof.

\subsection{Proof of Theorem~\ref{theo:sc}  }  \label{sec:rlProof}

\begin{lemma}\label{lem:rl1}
	Consider some integers $w\geq 1$ and $0\leq k\leq w-1$, a finite sequence $g_0, g_1, \ldots, g_{w-1}$ and  a number $a\in\Reals$ that satisfy :
\begin{enumerate}
  \item $\forall \ell\in\Nats \;\mif 0\leq \ell\leq w-2 \mthen g_{\ell+1}-g_{\ell}\geq 1$
  \item $\forall \ell \in \Nats \;\mif 0 \leq \ell \leq w - 3\mthen g_{\ell+2} - g_{\ell+1} \leq g_{\ell+1} - g_{\ell}$
  \item \mif $ k\leq w-2 \mthen a\geq g_{k+1}-g_k\melse a\geq 1$
  \item $\mif k\geq 1\mthen a\leq g_{k} - g_{k-1}$
\end{enumerate}

Define $f:[0,w)\to \Reals$ by $f(x)=g_{\lfloor x \rfloor}+x \mymod 1$ and  $h_: [0,w)\to \Reals$ by $h(x) = a(x - k) + g_k$

Then $h \geq f$.


\end{lemma}

\begin{proof}
First we show that
\begin{equation}\label{eq-jhkSDA}
  \forall \ell \in \lc 0,\ldots, w-1\rc, g_k -g_{\ell}\geq a (k-\ell)
\end{equation}
\noindent Case 1: $\ell <k$. Then
$
  g_k-g_{\ell}=\sum_{k'=\ell}^{k-1}(g_{k'+1}-g_{k'})
$. By 2) every term in the sum is $\geq g_{k}-g_{k-1}$, by 4) is also $\geq a$ and there are $(k-\ell)$ terms, this shows \eqref{eq-jhkSDA}.

\noindent Case 2: $\ell = k$. Then  \eqref{eq-jhkSDA} is obvious.

\noindent Case 3: $\ell > k$. Then
$
 g_{\ell}-g_k=\sum_{k'=k}^{\ell-1}(g_{k'+1}-g_{k'})
$. By 2) every term in the sum is $\leq g_{k+1}-g_{k}$; note that we must have $k\leq w-2$ thus by 3), every term in the sum is also $\leq a$; also, there are $\ell-k$ terms. Thus $g_{\ell}-g_k \leq a (\ell-k)$, which shows \eqref{eq-jhkSDA} in this case.

We now proceed with the proof of the lemma. Consider some arbitrary $x\in [0,w)$ and let $\ell=\lfloor x\rfloor$. Then
\begin{align}\label{eq-hjgsatq}
  f(x) &=x-\ell + g_{\ell}   \\
  h(x) &= a(x-\ell)+a(\ell-k)+g_k \\
  h(x)-f(x) &= \underbrace{(a-1)(x-\ell)}_{A} +\underbrace{g_k-g_{\ell}-a(k-\ell)}_{B}
\end{align}
Observe that we must have $a\geq 1$: if $k=w-1$ this follows from 3), and if $k\leq w-2$ it follows from 3) and 1); thus $A\geq 0$. Also $B\geq 0$ by \eqref{eq-jhkSDA}.

\end{proof}

\begin{lemma}\label{lem:rl3}

	Let $T>0$ and $P$ a bounded, wide-sense increasing function $[0,T)\to\Reals$. Extend $P$ to a function $\bar{P} \in \mathscr{F}$ by $\forall x\geq 0, \bar{P}(x) = \left\lfloor \frac{x}{T}\right\rfloor P(T^-)+ P(x \mymod T)$ where $P(T^-)\isdef\sup_{0\leq t <T }P(t)$.

 Also, consider an affine function $L$, defined by $L(x) = ax + b$ for some $a \geq \frac{P(T^-)}{T}$ and some $b\in \Reals$.

 If $L(x) \geq P(x)$ for all $x$ in $[0 , T)$ then $L \geq \bar{P}$.
\end{lemma}
%

\begin{proof}
	Observe that, for $x\geq 0$, $L(x) = a\left\lfloor \frac{x}{T}\right\rfloor T + L(x \mymod T)$. Now $L(x \mymod T) \geq P(x \mymod T)$ by hypothesis. Thus
  \begin{align}
    L(x) &\geq a\left\lfloor \frac{x}{T}\right\rfloor T + P(x \mymod T)  \\
     &\geq \frac{P(T^-)}{T}\left\lfloor \frac{x}{T}\right\rfloor T + P(x \mymod T)=\bar{P}(x)
  \end{align}

\end{proof}



\begin{lemma}\label{lem:rl4}
	 Let $f \in \mathscr{F}$ and a rate-latency function $\beta_{r,T}$ such that $r>0$, $T > 0$, and $\beta_{r,T} \leq f$. Assume that $\beta_{r,T}(x_1) = f(x_1)$ for $x_1 > T$.

 Then there is no other rate-latency function $\beta_{r',T'}$ (i.e., with $(r',T') \neq (r,T)$) such that $\beta_{r,T}\leq \beta_{r',T'}\leq f$.
\end{lemma}
\begin{proof}

	Assume that $\beta_{r,T}\leq \beta_{r',T'}\leq f$. The proof consists in showing that $(r,T)= (r',T')$.

First, we know that $\beta_{r,T}(x_1) = f(x_1)$ and $x_1 > T$; thus $r(x_1-T)=f(x_1)$ and
\begin{equation}\label{eq-jhkdsa}
  T=x_1 - \frac{f(x_1)}{r}
\end{equation}

Second, observe that we must have $T'\leq T$, since otherwise $\beta_{r,T}(T')>0= \beta_{r',T'}(T')$.

Third, observe that $f(x_1)=\beta_{r,T}(x_1)\leq \beta_{r',T'}(x_1)\leq f(x_1)$ thus $\beta_{r',T'}(x_1) = f(x_1)$ and
\begin{equation}\label{eq-jhkdsb}
T'=x_1 - \frac{f(x_1)}{r'}
\end{equation}

Combining the last three paragraphs, it follows $x_1 - \frac{f(x_1)}{r'} \leq x_1 - \frac{f(x_1)}{r}$, i.e., $r' \leq r$. Also, we must have $r'\geq r$, since otherwise $\forall x > x_0, ~\beta_{r,T}(x) > \beta_{r',T'}(x)$ with $x_0 = \frac{rT - r'T'}{r - r'}$. Thus, $r' = r$,  and it follows from~\eqref{eq-jhkdsa} and~\eqref{eq-jhkdsb} that $T' = T$.

\end{proof}

Now we proceed with the proof of Theorem~\ref{theo:sc}.

1) We first show that $r_k\leq r_{k+1}$ for $k=0...w_i-2$.
Define sequence $g$ by $g_k=\frac{1}{\lmin_i}\psi_i\lp k \lmin_i\rp$ for $k=0\ldots w_i-1$. By definition, we have  $g_{k+1} - g_k = $
\begin{equation}
			 1 + \frac{1}{\lmin_i}\sum_{j,j \neq i} \left( \min(k + 2 ,  w_j) - \min(k + 1 ,  w_j) \right) l^{max}_j
\end{equation}
Observe that $ \left( \min(k + 2 ,  w_j) - \min(k + 1 ,  w_j) \right)$ is equal to $1$ if $k + 1 <w_j$, and equal to $0$ otherwise. Thus, $g_{k+2} - g_{k+1}\leq g_{k+1} - g_k$ for $0 \leq k <w_i - 2$, which shows that  $r_k\leq r_{k+1}$ for $k=0...w_i-3$. Also, observe that $g_{k+1} - g_k \geq 1$, i.e., $r_k \leq 1$, for $0 \leq k \leq w_i - 2$. Hence, $r_{w_i - 2} \leq r_{w_i - 1}$.

2) Let $r\in [r^*_0,r^*_{k^*}]$ and let $T(r)$ be the value of $T$ defined in the Theorem, namely,
 $T(r)\isdef \psi_i(k\lmin_i) - \frac{k\lmin_i}{r}$, where $k$ is defined by $r^*_{k-1}\leq r <r^*_{k}$ if $r\in[r^*_0,r^*_{k^*})$ and $k=k^*$ if $r= r^*_{k^*}$. We now show that $\beta_{r, T(r)}\leq \gamma_i$.

 We consider two cases: $r^*_0 \leq r < r^*_{k^*} $ or $r = r^*_{k^*}$. For the former case, for any $r$, apply Lemma~\ref{lem:rl1} with $w = w_i$, $g$ as defined in 1), $k$ as defined in the paragraph above, and $a=\frac{1}{r}$. As by construction $\frac{1}{r_k}  < a \leq \frac{1}{r_{k-1}}$ and $\frac{1}{r_{k-1}} = g_{k } - g_{k-1}$, 3) and 4) are satisfied. For the latter case, apply again Lemma~\ref{lem:rl1} with the same $g$ and $w=w_i$ but now with $k = k^*$ and $a= \frac{1}{r} = \frac{1}{r^*_{k^*}}$. By construction, we have $ \frac{1}{r^*_{k^*}} \geq \frac{1}{r_{k^*}} = g_{k^*+1} - g_{k^*}$ and $ \frac{1}{r^*_{k^*}} \leq \frac{1}{r_{k^*-1}}=g_{k^*} - g_{k^*-1} $. Thus, conditions 3) and 4) of Lemma \ref{lem:rl1} are satisfied.  Let $f$ be the corresponding function $f$ in Lemma~\ref{lem:rl1}, i.e., $f(x)=g_{\lfloor x \rfloor}+x \mymod 1$ for $0\leq x<w_i$. Note that for both cases $f$ is the same. Also, let $f_{r}$ be the corresponding function $h$ in Lemma~\ref{lem:rl1}, i.e., $f_{r}(x) = \frac{1}{r}(x - k) + g_k$ for $0\leq x<w_i$. By Lemma~\ref{lem:rl1}, $f_{r}\geq f$.


 Observe that $f(w_i^-) = \frac{1}{\lmin_i} \lp \psi_i((w_i - 1) \lmin_i) + 1  \rp = \frac{1}{\lmin_i} \lp w_i \lmin_i + \sum_{j , j \neq i} w_j \lmax_j \rp = \frac{L_{\tot}}{\lmin_i} = \frac{w_i}{r^*}$. Then, as  $f_{r}(x)  \geq f(x)$ for $0 \leq x <w_i$ and $\frac{1}{r}  \geq  \frac{1}{r^*} =  \frac{f(w_i^-)}{w_i}$, we can apply Lemma \ref{lem:rl3} with $P = f$ and $L = f_{r}$. It gives us $\bar{f}$ defined by $\bar{f}(x) = \lfloor \frac{x}{w_i} \rfloor \frac{L_{\tot} }{\lmin_i}+ f\lp x \mymod w_i \rp  $ such that  $f_{r}  \geq \bar{f}$.

 Then, by using \eqref{lem:lsi2}, $f_{r}^{\downarrow}  \leq \bar{f}^{\downarrow}$. Also, as $\bar{f}^{\downarrow} \geq 0$, we have $\lb f_{r}^{\downarrow} \rb^+ \leq \bar{f}^{\downarrow}$. Note that for an increasing, linear function $L$, defined by $\forall x \geq0, L(x) = ax + b$ with some $a >0$ and $b > 0$, we have $\lb L^{\downarrow} \rb ^+ =\beta_{\frac{1}{a},b}$; and observe that $f_{r}(x) =  \frac{x}{r} + g_k - \frac{k}{r} = \frac{x}{r} + \frac{T(r)}{\lmin_i}$. Hence, $\lb f_{r}^{\downarrow} \rb ^+  = \beta_{r,\frac{T(r)}{\lmin_i}}$.

 Until now, we have shown that $\beta_{r,\frac{T(r)}{\lmin_i}} \leq \bar{f}^{\downarrow}$. Lastly, we show that $\lmin_i\bar{f}^{\downarrow}(\frac{x}{\lmin_i}) = \gamma_i(x)$ and $\lmin_i  \beta_{r,\frac{T(r)}{\lmin_i}} (\frac{x}{\lmin_i}) = \beta_{r,T(r)}(x) $. Observe that $\lmin_i\bar{f}(\frac{x}{\lmin_i}) = \lfloor \frac{x}{w_i\lmin_i} \rfloor L_{\tot} + \psi_i( ( \frac{x}{\lmin_i}  \mymod w_i ) \lmin_i)$. Also, $\psi_i(x) = \lfloor \frac{x}{w_i\lmin_i} \rfloor L_{\tot} +\psi_i(x \mymod w_i\lmin_i)$. Hence, we have $ \psi_i(x) = \lmin_i\bar{f}(\frac{x}{\lmin_i}) $. By using Lemma \ref{lem:inverse2} with $l = m = \lmin_i$, $\lmin_i\bar{f}^{\downarrow}(\frac{x}{\lmin_i}) = \psi^{\downarrow}_i(x) = \gamma_i(x)$. Also, observe that $\lmin_i  \beta_{r,\frac{T(r)}{\lmin_i}} (\frac{x}{\lmin_i}) = \beta_{r,T(r)}(x) $.

  Combine the last paragraphs to conclude that $\beta_{r, T(r)}\leq \gamma_i$ for all $r$ in $[r^*_0,r^*_{k^*}]$.

3) We now show that for any $r\in [r^*_0,r^*_{k^*}]$, $\beta_{r,T(r)}$ is a non-dominated lower-bound of $\gamma_i$. Let $r'\geq 0, T'\geq 0$ such that $\beta_{r,T(r)}\leq \beta_{r',T'}\leq \gamma_i$. We have to show that $r'=r$ and $T'=T(r)$.

First, if $r$ in $[r^*_0,r^*_{k^*})$, observe that $\beta_{r,T(r)}(x) =  \gamma_i(x)$ for $x = \psi_i(k\lmin_i) > \psi_i(k\lmin_i)  - \frac{k\lmin_i}{r} = T(r)$. Then, apply Lemma \ref{lem:rl4} with $\beta_{r,T} = \beta_{r,T(r)}$ and $f =  \gamma_i$ to conclude that $r'=r$ and $T'=T(r)$.

Second, if $r = r^*_{k^*}$, observe that $\beta_{r,T(r)}(x) =  \gamma_i(x)$ for $x = \psi_i(k^*\lmin_i) + L_{\tot} > T(r)$. Again, apply Lemma \ref{lem:rl4} with $\beta_{r,T} = \beta_{r,T(r)}$ and $f =  \gamma_i$ to conclude that $r'=r$ and $T'=T(r)$.

4) We now show that there is no other non-dominated rate-latency function, $\beta_{r',T'}$, that is upper bounded by $\gamma_i$.

First, we must have $T' \geq T(r^*_0)$. This is because $\gamma_i(x) = 0$ for $x \leq \psi_i(0) = T(r^*_0)$.

Second, we must have $r' \geq r^*_0$. Otherwise, we have $r' < r^*_0$ and we previously showed $T' \geq T(r^*_0)$. Thus, $\beta_{r',T'} \leq \beta_{r^*_0 , T(r^*_0)} \leq \gamma_i$, which is in contradiction with $\beta_{r',T'} $ being non-dominated.

Third, we must have $r' \leq  r^*_{k^*}$. We proceed to prove this by contradiction. If $T' \geq T(r^*_{k^*})$ and $r' >  r^*_{k^*}$, observe that $\beta_{r',T'}(x_0) = \beta_{r_{k^*}^*,T(r^*_{k^*})}(x_0)$ with $x_0 = \frac{r'T'+ r_{k^*}^* T(r^*_{k^*}) }{r' - r_{k^*}^* }$ and $\forall x, x > x_0 \Rightarrow \beta_{r',T'}(x) > \beta_{r_{k^*}^*,T(r^*_{k^*})}(x)$; for any arbitrary, non-negative integer $k$, let $x_k$ be defined by $x_k = \psi_i(k^*\lmin_i) + kL_{\tot}$. Then observe that $\beta_{r_{k^*}^*,T(r^*_{k^*})}(x_k) = \gamma_i(x_k)$. Choose some $k$ large enough such that $x_k > x_0$; then, $\beta_{r',T'}(x_{k}) > \beta_{r_{k^*}^*,T(r^*_{k^*})}(x_{k}) = \gamma_i(x_{k})$, which is in contradiction with $\beta_{r',T'} \leq \gamma_i$.
Also,  if $T' < T(r^*_{k^*})$ and $r' > r^*_{k^*}$, we have  $\forall x, x > T' \Rightarrow \beta_{r',T'}(x) > \beta_{r_{k^*}^*,T(r^*_{k^*})}(x)$. Choose some $k$ large enough such that $x_k > T'$; then, $\beta_{r',T'}(x_{k}) > \beta_{r_{k^*}^*,T(r^*_{k^*})}(x_{k}) = \gamma_i(x_{k})$, which is in contradiction with $\beta_{r',T'} \leq \gamma_i$.  Therefore,  $r' >  r^*_{k^*}$ is in contradiction with $\beta_{r',T'} \leq \gamma_i$.


Therefore, we must have $r'$ in $[r^*_0 , r^*_{k^*}]$. We now show that $T'= T(r')$. Because otherwise, if $T' < T(r') $, we have $\beta_{r',T(r')} \leq \beta_{r',T'}  \leq \gamma_i$ which is in contradiction with $\beta_{r',T(r')} $ being a non-dominated rate latency function. Also, if $T' > T(r')$, we have $\beta_{r',T'} \leq \beta_{r',T(r')}   \leq \gamma_i$, which is in contradiction with $\beta_{r',T'} $ being non-dominated.

\subsection{Tightness Proofs}\label{app:1}
We use the following Lemma about the lower pseudo-inverse technique.

\begin{lemma}\label{lsi:lem}
For a right-continuous function $f$ in $\mathscr{F}$ and $x,y$ in $\mathbb{R^+}$, $f^{\downarrow}\left(y \right) = x$ if and only if
$
		f(x) \geq y $ and there exists some $\varepsilon>0$ such that $ \forall x' \in (x-\varepsilon, x), f(x') < y
$.
\end{lemma}

\begin{proof}
	
$\Rightarrow$: 

Let $S=\lc x', f(x')\geq y\rc$ so that $x=\inf S$ \eqref{def:lsi}. From the definition of an $\inf$, there exists a sequence $x_n$ such that $x_n\in S$ for all $n$, $x_n\geq x$, and $\lim_{n\to \infty}x_n=x$. Since $f$ is right-continuous, $\lim_{n\to \infty}f(x_n)=f(x)$, which shows that $f(x)\geq y$. Also, again by definition of an $\inf$, any $x'<x$ does not belong to $S$, i.e. $\forall x' < x, f(x') < y$.
	
$\Leftarrow$:

By the first part of the hypothesis, $x\in S$ therefore $x\geq \inf S = f^{\downarrow}\left(y \right)$.
Let also $S'=\lc x', f(x')< y\rc$ so that $f^{\downarrow}\left(y \right)=  \sup S'$ \eqref{def:lsi}. By the second part of the hypothesis, $S'$ contains the interval $(x-\varepsilon,x)$ hence $\sup S'\geq x$, which shows that $f^{\downarrow}\left(y \right)\geq x$. Combining the two shows that $f^{\downarrow}\left(y \right)= x$.

\end{proof}

\begin{proof}[Proof of Theorem \ref{thm:tight}]
%
We prove that, for any value of the system parameters, for any $\tau >0$, and for any flow $i$, there exists one trajectory of a system such that
\begin{equation}\label{eqn;tight}
 \begin{aligned}
	&\exists s \geq 0, \, (s,s+\tau] \, \text{ is backlogged for flow $i$}\\ &\text{and }R_i^*(s + \tau) - R_i^*(s) = \beta_i(\tau)
	\end{aligned}
	\end{equation}

	\textbf{Step 1: Constructing the Trajectory}

    1) Flows are labeled in order of weights, i.e., $w_j\leq w_{j+1}$.
    	
	2) At time $0$, the input of every queue $j \neq i$ is a burst of size 
 $\left \lceil \frac{ \beta(\tau)}{\lmax_j} \right  \rceil \lmax_j + w_j\lmax_j $ .
	
	3) Every flow, $j\neq i$, is packetized according to its maximum packet size, $\lmax_j $.

    4) The output of the system is at rate $K$ (the Lipschitz constant of $\beta$) from time $0$ to times $s$, which is defined as the time at which queue $i$ is visited at cycle $w_i$ in the first round, namely
    \begin{equation}\label{eq-kjlsda}
      s=\frac{1}{K}\sum_{j , j\neq i} \min \lp w_i-1 , w_j \rp\lmax_j
    \end{equation}
    It follows that  \begin{equation}\label{eq-asfdlsda}\forall t\in[0,s], R^*(t)=Kt\end{equation}
	
	5) The input of queue $i$ starts just after time $s$, with a burst of size $\left \lceil \frac{ \beta(\tau)}{\lmin_i}\right  \rceil \lmin_i$.
	
	6) Flow $i$ is packetized according to its minimum packet size, $\lmin_i$.
	
	7) After time $s$, the output of the system is equal to the guaranteed service; by 2) and 5), the busy period lasts for at least $\tau$, i.e.,
\begin{equation}\label{eq-jhgdsf}
  \forall t\in [s,s+\tau], R^*(t)=R^*(s)+\beta(t-s)
\end{equation}  In particular,
	\begin{equation}
\label{eqn:3}
	R^*(s + \tau) - R^*(s) = \beta(\tau)
	\end{equation}
	
	If we apply $\psi^{\downarrow}_i$ to both sides of \eqref{eqn:3}, the right-hand side is equal to $\beta_i(\tau)$. Thereby, we should prove:
	\begin{equation} \label{eqn:2}
 \psi^{\downarrow}_i\left(R^*(s + \tau) - R^*(s) \right) =\ R_i^*(s +  \tau) - R_i^*(s)
	\end{equation}
Let $y=R^*(s+\tau)-R^*(s)$ and $x=R_i^*(s+\tau)-R_i^*(s)$. Our goal is now to prove that
	\begin{equation} \label{eqn:2a}
 \psi^{\downarrow}_i\left(y\right) =x
	\end{equation}

	From 5), we know that the first packet of flow $i$ is served at the first cycle of a round ($C = 1$ in Algorithm \ref{alg:IWRR}). Thus, applying Lemma \ref{lem:NVg2} and (P5) in Lemma \ref{lem:numg}, the number of services to each flow $j$ is equal to $ \phi_{i,j}(p)$. From 2), flow $j$ sends packets with the maximum length. Thus:
	\begin{equation} \label{eqn:1}
	\sum_{j,j \neq i} R_j^*(s + \tau_{\sigma(p)}) - R_j^*(s)   =  \sum_{j,j \neq i} \phi_{i,j}(p)\lmax_j
	\end{equation}

Now there are two cases for $s + \tau$ (\ref{ssec:t}).

	\textbf{Case 1:} $s + \tau < \tau_{\sigma(p)}$
In this case the scheduler is not serving flow $i$ in $[\tau_{\sigma(p)} , s + \tau]$ and $x=p\lmin_i$. Thus $R_i^*(s + \tau) = R_i^*(\tau_{\sigma(p)})$. 
It follows that
\begin{equation}
	\begin{aligned}
	&\psi_i(x) = x + \underbrace{\sum_{j,j\neq i} \phi_{i,j}(\lfloor \frac{x}{\lmin_i} \rfloor)\lmax_j    }_{\sum_{j,j\neq i} R_j^*(\tau_{\sigma(p)}) - R_j^*(s)}\\
	& y = x + \sum_{j,j\neq i} R_j^*(s + \tau) - R_j^*(s)
	\end{aligned}
	\end{equation}
and thus
	\begin{equation}
\label{eq:jhkds}
\psi_i(x) \geq y
	\end{equation}
Let $x-\lmin_i<x'<x$; flow $i$'s output becomes equal to $x'$ during the emission of packet $p-1$ thus
	\begin{equation}
	\psi_i(x') = x' + \sum_{j,j\neq i} R_j^*(\tau_{\sigma(p-1)}) - R_j^*(s)\\
	\end{equation}
Hence
\begin{equation}
\label{eq:jhkdt}
\forall x'\in (x-\lmin_i,x),\psi_i(x') < y
\end{equation}
Combining \eqref{eq:jhkds} and \eqref{eq:jhkdt} with Lemma~\ref{lsi:lem} shows \eqref{eqn:2a}.

	\textbf{Case 2:} $s + \tau \geq \tau_{\sigma(p)}$
In this case the scheduler is serving flow $i$ in $[\tau_{\sigma(p)} , s + \tau]$.		For every other flow $j$, we have $R_j^*(s + \tau) = R_j^*( \tau_{\sigma(p)})$. Hence,
\begin{equation}\label{eq:kjlasdf}
\psi_i(x)=
 R_i^*(s +  \tau) - R_i^*(s) +  \sum_{j,j \neq i} \phi_{i,j}(p)\lmax_j=y
\end{equation}
As with case 1, for any $x'\in ((p-1)\lmin_i,x)$, we have $\psi_i(x)<y$, which shows \eqref{eqn:2a}.

%
%
This shows that \eqref{eqn;tight} holds. It remains to show that the system constraints are satisfied.

	\textbf{Step 2: Verifying the Trajectory}
%
We need to verify that the service offered to the aggregate satisfies the strict service curve constraint.
Our trajectory has one busy period, starting at time $0$ and ending at some time $T_{\max}\geq \tau$. We need to verify that
\begin{equation}\label{eq-ljkhsdf}
  \forall t_1, t_2\in [0,T_{\max}] \mwith t_1<t_2, R^*(t_2)-R^*(t_1)\geq \beta(t_2-t_1)
\end{equation}

\textbf{Case 1:} $t_2<s$

Then $R^*(t_2)-R^*(t_1)=K(t_2-t_1)$. Observe that, by the Lipschitz continuity condition on $\beta$, for all $t\geq 0$, $\beta(t)=\beta(t)-\beta(0)=\beta(t)\leq Kt$ thus $K(t_2-t_1)\geq \beta(t_2-t_1)$.

\textbf{Case 2:} $t_1<s\leq t_2$
Then $R^*(t_2)-R^*(t_1)=\beta(t_2-s)+K(s-t_1)$. By the Lipschitz continuity condition:

\begin{equation}\label{eq-jhkasd}
  \beta(t_2-t_1)-\beta(t_2-s)\leq K(s-t_1)
\end{equation}
thus $R^*(t_2)-R^*(t_1)\geq\beta(t_2-t_1)$.

\textbf{Case 3:} $s\leq t_1< t_2$
Then $R^*(t_2)-R^*(t_1)=\beta(t_2)-\beta(t_1)\geq \beta(t_2-t_1)$ because $\beta$ is super-additive.

\end{proof}

\begin{proof}[Proof of Theorem \ref{thm:tightWRR}]
	The proof is very similar to the proof of Theorem \ref{thm:tight}. The necessary changes in the proof are the following:
	
	1) $s$ is the time of the first visit to flow $i$.
	
	2) Instead of functions $\psi_i$ and $\phi_{i,j}$, use
        functions $\psi_i'$ and $\phi_{i,j}'$, defined in \eqref{eq:psi'} and \eqref{eq:phi'}.
\end{proof}

\begin{proof}[Proof of Theorem \ref{thm:dbIWRR}]
	The proof contains the following steps:
	
	1)	Consider the same trajectory as in the proof of Theorem \ref{thm:tight}, yet with one difference: the input of flow $i$ is $R_i(t) = \alpha_i(t-s)$ for $t \geq s$ and zero before $s$. Observer that as $\alpha_i$ is sub-additive, $\forall t_1,t_2:~ t_2 \geq t_1 \geq s \Rightarrow R_i(t_2) -   R_i(t_1) =  \alpha_i(t_2) - \alpha_i(t_1) \leq \alpha_i(t_2 - t_1)$.
	
	2) Define $s' = \inf\{ u > 0 |  \alpha_i(u) \leq \beta_i(u)\}$. This is the first time after zero that the service curve meets the arrival curve.  Note that $s'$ can be infinite as well.
	
	3) Then, it is guaranteed that flow $i$ is backlogged in $(s , s + s']$. Therefore, using \eqref{eqn;tight}, we have $R_i^*(t) = \beta_i(t-s)$ for $t \geq s$ and zero before $s$.
	
	4) Combining 1 and 3, the horizontal deviation of $R_i$ and $R^*_i$ in $(s , s + s']$ is equal to the  horizontal deviation of $\alpha_i$ and $\beta_i$ in $[0 ,s']$.
	
	4) Using \cite[Sec. 5.3.3]{bouillard_deterministic_2018}, the horizontal deviation of $\alpha_i$ and $\beta_i$ can be restricted to $[0 ,s']$.
	
	Thereby, we find a valid trajectory (verified in the proof of Theorem \ref{thm:tight}) where the delay bound is achieved.
\end{proof}

\begin{proof}[Proof of Theorem \ref{thm:dbWRR}]
	The same proof of Theorem \ref{thm:dbIWRR} works here as well. However, we use the trajectory defined in the proof of Theorem \ref{thm:tightWRR}.
\end{proof}

	\fi

%
%

\end{document}